\documentclass[sigconf]{aamas}  % do not change this line!

\AtBeginDocument{%
  \providecommand\BibTeX{{%
    \normalfont B\kern-0.5em{\scshape i\kern-0.25em b}\kern-0.8em\TeX}}}
    
\usepackage{booktabs}    
\usepackage{flushend} % do not change this line!
\usepackage{helvet}
\usepackage{courier}
\usepackage{hyperref}       % hyperlinks
\usepackage{url}            % simple URL typesetting
\usepackage{graphicx}    
\usepackage{amsmath}
\usepackage{float}
\usepackage{subfig}
\usepackage{booktabs}
\usepackage{algorithm, float}
\usepackage{algorithmic}
\usepackage{amsthm}
\usepackage{mathrsfs} 
\usepackage{amsfonts}

\usepackage{mathtools, nccmath}
\usepackage{tikz}
\usetikzlibrary{positioning, shapes, patterns}

\usepackage{bm}
\usepackage{mathrsfs}
\DeclareMathOperator{\E}{\mathbb{E}}

\usepackage[utf8]{inputenc}
\usepackage[english]{babel}

\setcopyright{ifaamas}  % do not change this line!
\copyrightyear{2020} % do not change this line!
\acmYear{2020} % do not change this line!
\acmDOI{} % do not change this line!
\acmPrice{} % do not change this line!
\acmISBN{} % do not change this line!
\acmConference[AAMAS'20]{Proc.\@ of the 19th International Conference on Autonomous Agents and Multiagent Systems (AAMAS 2020)}{May 9--13, 2020}{Auckland, New Zealand}{B.~An, N.~Yorke-Smith, A.~El~Fallah~Seghrouchni, G.~Sukthankar (eds.)}  % do not change this line!
\newtheorem{lemm}{Lemma}
\begin{document}

\title{Multi Type Mean Field Reinforcement Learning}  % put your title here!

% AAMAS: as appropriate, uncomment one subtitle line; see camera ready instructions
%\subtitle{Extended Abstract}
%\subtitle{Blue Sky Ideas Track}
%\subtitle{JAAMAS Track}
%\subtitle{Doctoral Consortium}                              
%\subtitle{Demonstration}
%\subtitlenote{Please refrain from using subtitle notes}

% replace this with your author block!
\author{Sriram Ganapathi Subramanian}
\authornote{Sriram did this work while he was an intern at Borealis AI.}
\affiliation{%
 \institution{University of Waterloo}
 \streetaddress{200 University Ave. West}
 \city{Waterloo} 
 \state{Ontario} 
 \postcode{N2L 3G1}
}
\email{s2ganapa@uwaterloo.ca}

\author{Pascal Poupart}
\affiliation{%
 \institution{Borealis AI}
 \streetaddress{420 Wes Graham Way}
 \city{Waterloo} 
 \state{Ontario} 
 \postcode{N2L 0J6}
}
\email{pascal.poupart@borealisai.com}

\author{Matthew E.\ Taylor}
\affiliation{%
 \institution{Borealis AI}
 \streetaddress{10020 101A Ave NW}
 \city{Edmonton} 
 \state{Alberta}
 }
\email{matthew.taylor@borealisai.com}

\author{Nidhi Hegde}
\affiliation{%
 \institution{Borealis AI}
 \streetaddress{10020 101A Ave NW}
 \city{Edmonton} 
 \state{Alberta}
 }
\email{nidhi.hegde@borealisai.com}

%% The example's default list of authors is too long for headers
\renewcommand{\shortauthors}{Ganapathi Subramanian et al.}

\begin{abstract}
 Mean field theory provides an effective way of scaling multiagent reinforcement learning algorithms to environments with many agents that can be abstracted by a virtual mean agent.  In this paper, we extend mean field multiagent algorithms to multiple types. The types enable the relaxation of a core assumption in mean field reinforcement learning, which is that all agents in the environment are playing almost similar strategies and have the same goal. We conduct experiments on three different testbeds for the field of many agent reinforcement learning, based on the standard MAgents framework. We consider two different kinds of mean field environments: a) Games where agents belong to predefined types that are known a priori and b) Games where the type of each agent is unknown and therefore must be learned based on observations.  We introduce new algorithms for each type of game and demonstrate their superior performance over state of the art algorithms that assume that all agents belong to the same type and other baseline algorithms in the MAgent framework. 
\end{abstract}

% AAMAS: the ACM CCS are encouraged but optional within AAMAS papers
%\begin{CCSXML}
%<ccs2012>
%  <concept_id>10010520.10010553.10010562</concept_id>
%  <concept_desc>Computer systems organization~Embedded systems</concept_desc>
%  <concept_significance>500</concept_significance>
%  <concept_id>10010520.10010575.10010755</concept_id>
%  <concept_desc>Computer systems organization~Redundancy</concept_desc>
%  <concept_significance>300</concept_significance>
%  <concept_id>10010520.10010553.10010554</concept_id>
%  <concept_desc>Computer systems organization~Robotics</concept_desc>
%  <concept_significance>100</concept_significance>
%  <concept_id>10003033.10003083.10003095</concept_id>
%  <concept_desc>Networks~Network reliability</concept_desc>
%  <concept_significance>100</concept_significance>
%</ccs2012>
%\end{CCSXML}
%
%\ccsdesc[500]{Computer systems organization~Embedded systems}

 % put your comma-separated keywords here!

\maketitle

%%%%%%%%%%%%%%%%%%%%%%%%%%%%%%%%%%%%%%%%%%%%%%%%%%%%%%%%%%%%%%%%%%%%%%%%%%%%%%%%%%%%%%%%%%%%%%%%%%%%%%%%%
%% start of main body of paper

\section{Introduction}

Multiagent reinforcement learning (MARL) is a quickly growing field with lots of recent research pushing its boundaries \cite{hernandez2019survey,nguyen2018deep}. Yet, scaling the multiagent algorithms to environments with a large number of learning agents continues to be a problem \cite{bucsoniu2010multi}. Research advances in the field of MARL \cite{bloembergen2015evolutionary, bu2008comprehensive} deal with only a limited number of agents and the proposed methods cannot be easily generalized to more complex scenarios with many agents. Some recent research has used the concept of mean field theory for enabling the use of MARL approaches to environments with many agents \cite{mguni2018decentralised,pmlr-v80-yang18d}. Yet, the current algorithms implemented for many agents require some strong assumptions about the game environment to perform adequately. The important mean field approximation reduces a many agent problem into a simplified two agent problem where all the other participating agents are approximated as a single mean field. This mean field approximation, however, would be valid only for scenarios where all the agents in the environment can be considered similar to each other in objectives and abilities. Real world applications often have a set of agents that are diverse, and therefore it is virtually  impossible to aggregate them into a single mean field. 

In this paper, we introduce a concept of multiple types to model agent diversity in the mean field approximation for many agent reinforcement learning. The types are groupings applied to other agents in the environment in such a way that all members of a particular type play approximately similar strategies and have similar overall goals. Now, the modelling agent can consider each type to be a distinct agent, which has to be modelled separately. 
Within each type, the mean field approximation should still be reasonable as the agents within one particular type are more related to each other than other agents from different types. Thus, the many agent interaction is effectively reduced to $M$ agent interactions where $M$ is the number of types. Note that this is more complex than the simple two agent interaction considered by previous research and this can approximate a real world dynamic environment in a much better way. Most real world applications for many agent reinforcement learning can be broadly classified into two categories. The first category involves applications where we have predefined types and the type of each agent is known a priori. Some common applications include games with multiple teams (like quiz competitions), multiple party elections with coalitions, airline price analysis with multiple airlines forming an alliance beforehand, etc. We call these applications predefined {\em known} type scenarios. The other category are applications that involve a large number of agents that may have different policies due to differences in their rewards, actions or observations.  Common examples are demand and supply scenarios, stock trading scenarios, etc., where one type of agent may be risk averse while another type may be risk seeking. We call these applications pre-defined {\em unknown} type scenarios, since there are no true underlying types like the previous case and a suitable type therefore must be assigned through observations. Another important aspect in this paper will be the notion of {\em neighbourhood} since each individual agent may be impacted more by agents whose states are ``closer'' (according to some distance measure) to the state of an agent.  For instance, in battle environments, nearby agents pose a greater threat than far away agents.

Using the open source test environment for many agent games --- MAgents \cite{zheng2018magent} --- we consider three testbeds that involve many strategic agents in the environment. Two of these testbeds correspond to the known type and the third one corresponds to the unknown type. We introduce two different algorithms for the known and unknown type scenarios. We demonstrate the superior performance of these algorithms in comparison to previous algorithms that assume a single type of agents in the testbeds.

\section{Background}
\textbf{Reinforcement Learning}: Single agent reinforcement learning (RL) \cite{sutton1998introduction} is the most common form of reinforcement learning in the literature \cite{arulkumaran2017brief}. Here the problem is modelled in the framework of a Markov Decision Process (MDP) where the MDP is composed of $\langle \mathcal{S},A,P,R \rangle$, where $\mathcal{S}$ denotes the set of states that the agent can move into, $A$ denotes the set of actions that the agent can take, $P$ denotes the transition distribution ($P(s'|s,a)$) and
$R$ denotes the reward function ($R(s,a)$). The agents are allowed to explore the environment during the process of training and collect experience tuples $\langle s,a,s',r\rangle$. An agent learns a policy $\pi: \mathcal{S} \rightarrow A$, which is a mapping from states to actions where the goal is to maximize the expected cumulative rewards $\sum_t \gamma^t R(s_t,a_t)$, where $\gamma \in [0,1)$ is the discount factor. An optimal policy obtains the highest cumulative rewards among all possible policies. 

In Multiagent reinforcement learning (MARL), there is a notion of stochastic games \cite{bowling2000analysis}, where the state and action space are defined as Cartesian products of individual states and actions of different agents in the environment. A stochastic game  can be considered a special type of normal form game \cite{jordan1991bayesian}, where a particular iteration of the game depends on previous game(s) played and the experiences of all the agents in the previous game(s). A stochastic game can be defined as a tuple $\langle \mathcal{S},N,{\textbf{A}},P,R \rangle$ where $\mathcal{S}$ is a finite set of states (assumed to be the same for all agents), $N$ is a finite set of $n$ agents, ${\textbf{A}} = A^1 \times \cdots \times A^n $ where $A^j$ denotes the set of actions of the agent $j$. $P$ is the transition distribution $P(s'|s,\textbf{a})$ where $\textbf{a}= (a^1, \ldots, a^n)$ and $R_j(s,\textbf{a}) = r^j$ is the reward function with $r^j$ denoting the reward received by agent $j$. Here each agent is trying to learn a policy that maximizes its return upon consideration of opponent behaviours. Agents are typically self interested and the combined system moves towards a Nash equilibrium \cite{maskin1999nash}.  Scalability in MARL environments is often a bottleneck. Important research efforts are aimed at handling only up to a handful of agents and the solutions or algorithms considered become intractable in many agent scenarios.

\textbf{Mean Field Reinforcement Learning}: Mean field theory approximates many agent interactions in a multiagent environment into two agent interactions \cite{lasry2007mean} where the second agent corresponds to the mean effect of all the other agents. 
This allows domains with many agents that were previously considered intractable to be revisited and scalable approximate solutions to be devised \cite{pmlr-v80-yang18d,mguni2018decentralised}. 
Mean field reinforcement learning applies mean field approximation to stochastic games \cite{pmlr-v80-yang18d}.
In the paper by Yang et al.~\shortcite{pmlr-v80-yang18d}, the multi-agent $Q$ function for stochastic games is decomposed additively into local $Q$ functions that capture pairwise interactions: 
\begin{equation}
Q^{j}(s,\textbf{a}) = \frac{1}{n^j}\sum_{k \in \eta(j)} Q^j (s, a^j, a^k). 
\end{equation}
Here $n^j$ is the number of neighbours of the agent $j$ and $\eta(j)$ 
is the index set of neighbouring agents. Yang et al.~\shortcite{pmlr-v80-yang18d} showed that this decomposition can be well approximated by the mean field $Q$ function $Q^j(s,\textbf{a})\approx Q^j_{MF}(s,a^j, \overline{a}^j)$ under certain conditions.  The mean action $\overline{a}^j$ on the neighbourhood $\eta(j)$ of agent $j$ is expressed as $\overline{a}^j = \frac{1}{n^j} \sum_{k\in \eta(j)} a^k$ where $a^k$ is the action of each neighbour $k$.  In the case of discrete actions, $a^k$ is a one-hot vector encoding and $\overline{a}^j$ is a vector of fractions corresponding to the probability that each action may be executed by an agent at random.

The mean field $Q$ function can be updated in a recurrent manner as follows,

\begin{equation}
\label{eq:MFQ}
\medmath{Q^j_{t + 1}(s,a^j, \overline{a}^j)  =  (1-\alpha) Q_t^j(s,a^j, \overline{a}^j) + \alpha[r^j + \gamma v^j_t(s')]} 
\end{equation}

\noindent where $r^j$ is the reward obtained. The $s,s'$ are the old and new states respectively. $\alpha_t$ is the learning rate. The value function $v^j_t(s')$ for agent $j$ at time $t$ is given by, 

\begin{equation}
\medmath{v^{j}_{t}(s') = \sum_{a^j}\pi^j_t(a^j|s',\overline{a}^j)  \E_{a^{-j}_{i}\sim \pi_{i,t}^{-j}} Q_t^j(s',a^j, \overline{a}^j)}.
\end{equation}

\noindent Here, the term $\overline{a}^{j}$ denotes the mean action of all the other agents apart from $j$. 
The mean action for all the agents is first calculated using the relation, $ \overline{a}^j = \frac{1}{N^j} \sum_{k} a^k, a^k \sim \pi_{t}^k(\cdot|s,\overline{a}^k_{-})
 $,  where $\pi_t^k$ is the policy of the agent $k$ at time $t$, and $\overline{a}^k_{-}$ represents the previous mean action for the neighbours of the agent $k$. $N^j$ denotes the total number of agents in the neighbourhood of the agent $j$.
Then, the Boltzmann policy for each agent $j$ is calculated using $\beta$, which is the Boltzmann softmax parameter,
\begin{equation}
\medmath{\pi_t^j(a^j|s, \overline{a}^j) = \frac{exp(-\beta Q^j_t(s,a^j, \overline{a}^j))}{\sum_{a^{j'}\in A^j}exp(-\beta Q_t^j(s,a^{j'},\overline{a}^{j}))}}.
\end{equation}

\section{Mean Field MARL and Types}

We consider environments where there are $M$ types that the neighbouring agents can be classified into. In this paper, we assume that the $Q$ function decomposes additively according to a partition of the agents into $X^j$ subsets that each include one agent of each type. This decomposition can be viewed as a generalization of pairwise decompositions to multiple types since each term depends on a representative from each type. 

   Let the standard multi-agent $Q$ function be $Q^j(s,\boldsymbol{a})$, then, 

\begin{equation}\label{eq:meanfield}
Q^j(s, \boldsymbol{a}) = \frac{1}{X^j} \sum_{i = 1}^{X^j}[Q^j(s, a^j, a^{k_i}_1, a^{k_i}_2, \ldots, a^{k_i}_M) ].    
\end{equation}

\noindent Here we have a total of $M$ types and $a^{k}_m$ denotes the action of agent $k$ belonging to type $m$ in the neighbourhood of agent $j$. Notice that this representation of the $Q$ function includes the interaction with each one of the types and is not a simple pairwise interaction as done by \cite{pmlr-v80-yang18d}. Let us assume that we have a scheme in which we can classify each agent into one of these subsets. Note that we do not need each group to contain an equal number of agents as we can always make a new subset that contains one agent of a type and other agents to be placeholder agents (dead agents) of other types. We will relax this requirement of decomposition shortly and in practice we do not need to make these subsets at all.  

\subsection{Mean Field Approximation}

We also assume discrete action spaces and use a one hot representation of the actions as in Yang et al.~\cite{pmlr-v80-yang18d}. The one hot action of each agent $k$ belonging to type $m$ in the neighbourhood of agent $j$ is represented as $ a^{k_m}_m = \overline{a}^j_m + \hat{\delta} ^{j,k_m}$ where $\overline{a}^{j}_{m} $ is the mean action of all agents in the neighbourhood of agent $j$ belonging to type $m$ and $\hat{\delta}^{j,k_m}$ is the deviation between the action of an agent and the mean action of its type. 

Let $\delta^{j,k_i} = [\hat{\delta}^{j,k_1}; \hat{\delta}^{j,k_2}; \cdots; \hat{\delta}^{j,k_M}]$ be a vector obtained by the concatenation of all such deviations of the agents in the neighbourhood of the agent $j$ belonging to each of the $M$ types (all agents of a single subset). Similar to \cite{pmlr-v80-yang18d}, we apply Taylor's theorem to expand the $Q$ function in Equation~\ref{eq:meanfield} to get,

\begin{equation}
   \begin{array}{l} 

Q^j(s,\textbf{a}) = \frac{1}{X^j} \sum_{i=1}^{X^j} Q^j(s, a^j, a^{k_i}_1, a^{k_i}_2, \ldots, a^{k_i}_M) \\
 = \frac{1}{X^j}\sum_{i=1}^{X^j} [Q^j(s, a^j, \overline{a}^j_1, \ldots , \overline{a}^j_M) \\ 
 \quad +  \nabla_{\overline{a}^j_1, \ldots, \overline{a}^j_M}Q^j(s, a^j, \overline{a}^j_1, \ldots, \overline{a}^j_M) \cdot \delta^{j,k_i} \nonumber \\ 
 \quad  + \frac{1}{2} \delta^{j,k_i} \cdot \nabla^2_{\Tilde{a}^j_1, \ldots, \Tilde{a}^j_M} Q^j(s, a^j, \Tilde{a}^j_1,\ldots , \Tilde{a}^j_M) \cdot \delta ^{j,k_i}] \nonumber \\
 = Q^j(s, a^j, \overline{a}^j_1, \ldots , \overline{a}^j_M) \\ 
 \quad + \nabla_{\overline{a}^j_1, \ldots, \overline{a}^j_M}Q^j(s,a^j, \overline{a}^j_1, \overline{a}^j_2, \ldots, \overline{a}^j_M) \cdot [\frac{1}{X^j}\sum_{i=1}^{X^j} \delta^{j,k_i}] \nonumber \\ 
 \quad + \frac{1}{2X^j} \sum_{i=1}^{X^j}[\delta ^{j,k_i} \cdot \nabla^2_{\Tilde{a}^j_1, \ldots, \Tilde{a}^j_M} Q^j(s,a^j, \Tilde{a}^j_1, \Tilde{a}^j_2, \ldots, \Tilde{a}^j_M) \cdot \delta^{j,k_i} ] \nonumber \\
 \medmath{= Q^j(s, a^j, \overline{a}^j_1, \ldots , \overline{a}^j_M) +  \frac{1}{2X^j} \sum_{i=1}^{X^j}[R^j_{s,a^j}(a^{k_i})] 
 \approx Q^j(s, a^j, \overline{a}^j_1, \ldots , \overline{a}^j_M)} \nonumber
\end{array}
\end{equation}

\noindent where $R^j_{s,a^j}(a^k) \triangleq \delta^{j,k} \cdot \nabla^2_{\Tilde{a}^j_1, \ldots, \Tilde{a}^j_M} Q^j(s,a^j, \Tilde{a}^j_1, \Tilde{a}^j_2, \ldots, \Tilde{a}^j_M) \cdot \delta^{j,k}  $

\noindent
which is the Taylor polynomial remainder. The summation term $ [\frac{1}{X^j}\sum_{i=1}^{X^j} \delta^{j,k_i}] $ sums out to 0. Finally, if we ignore the remainder terms $R^j_{s,a^j}$, we obtain the following approximation,

\begin{equation}\label{eq:MTMFapprox}
Q^{j}(s,\textbf{a}) \approx Q^j_{\textit{MTMF}} (s, a^j, \overline{a}^j_1, \ldots , \overline{a}^j_M).
\end{equation}

The magnitude of this approximation depends on the deviation $\hat{\delta}^{j,k_m}$ between each action $a_m^{k_m}$ and its mean field approximation $\bar{a}_m^j$.  More precisely, we can quantify the overall effect of the mean field approximation by the average deviation $\sum_k ||\hat{\delta}_k||_2 / N$.  Theorems \ref{theorem:bound1} and \ref{theorem:bound2} show that the average deviation is reduced as we increase the number of types and Theorem \ref{theorem:mtmfq} provides an explicit bound on the approximation in Eq. \ref{eq:MTMFapprox} based on the average deviation.

\begin{theorem}\label{theorem:bound1}
When there are two types in the environment, but they have been considered to be the same type, the average deviation induced by the mean field approximation is bounded as follows,

\begin{equation}
    \frac{\sum_k || \hat{\delta}_k||_2}{N} \leq \frac{K_1}{N}\epsilon_1 + \frac{K_2}{N}\epsilon_2 + \frac{K_1}{N} \alpha_1 + \frac{K_2}{N} \alpha_2
\end{equation}

\noindent where $N$ denotes the total number of agents in the environment; $K_1$ and $K_2$ denote the total number of agents of types $1$ and $2$, respectively; $\epsilon_1$ and $\epsilon_2$ are bounds on the average deviation for agents of types 1 and 2 respectively. 
 
\begin{equation*}
\begin{array}{l}

 \frac{1}{K_1} (\sum_{k_1} ||a_{k_1} - \overline{a}_1 ||_2) \leq \epsilon_1; 
 \frac{1}{K_2} (\sum_{k_2} ||a_{k_2} - \overline{a}_2 ||_2) \leq \epsilon_2 
\end{array} 
\end{equation*}

Similarly, $\alpha_1$ and $\alpha_2$ denote the errors induced by using a single type (instead of two types) in the mean field approximation.

\begin{equation}
\begin{array}{l}
 \alpha_1 = ||\overline{a}_1 - \overline{a} ||_2; 
 \alpha_2 = ||\overline{a}_2 - \overline{a} ||_2
 \end{array} 
\end{equation}

Furthermore, $a_{k_1}$ denotes an action of an agent belonging to type $1$ and $a_{k_2}$ denotes an action of an agent belonging to type $2$. Here $\overline{a}$ denotes the mean field action of all the agents, $\overline{a}_1$ denotes the mean action of all agents of type $1$ and $\overline{a}_2$ denotes the mean action of all agents of type $2$.

\end{theorem}

\begin{proof}

Since we have considered every agent to belong to only one type, the deviation between the agent's action and the overall mean action is the error estimate. Hence, we will have the following, 

\begin{equation*}
    \begin{array}{l}
      \frac{\sum_k || \hat{\delta}_k||_2}{N} = \frac{\sum_k||a^j - \overline{a}^j ||_2}{N}
    \end{array}
\end{equation*}

\begin{equation*}
\begin{array}{l}
     \frac{\sum_k || \hat{\delta}_k||_2}{N} = \frac{1}{N} (\sum_{
k_1} ||a_{k_1} - \overline{a}|| + \sum_{k_2} ||a_{k_2} - \overline{a} ||).
     
\end{array}
\end{equation*}

The superscript ($j$) and subscript (2) have been dropped for simplicity.

\begin{equation}\label{eq:singletypebound}
    \begin{array}{l}
          \medmath{= \frac{1}{N} (\sum_{
k_1} ||a_{k_1} - \overline{a_1} + \overline{a}_1 - \overline{a}|| + \sum_{k_2} ||a_{k_2} - \overline{a_2} + \overline{a_2} - \overline{a} ||)}

\\

\medmath{\leq \frac{1}{N} (\sum_{
k_1} ||a_{k_1} - \overline{a_1}|| + \sum_{k_1} ||\overline{a}_1 - \overline{a}|| + \sum_{k_2} ||a_{k_2} - \overline{a_2}|| + \sum_{k_2} ||\overline{a_2} - \overline{a} ||)}

\\

\medmath{=  \frac{1}{N} (\sum_{
k_1} ||a_{k_1} - \overline{a_1}|| + K_1 ||\overline{a}_1 - \overline{a} || + \sum_{k_2} ||a_{k_2} - \overline{a_2}|| + K_2 ||\overline{a_2} - \overline{a} ||)}
\\

\medmath{\leq \frac{K_1}{N}\epsilon_1 + \frac{K_2}{N}\epsilon_2 + \frac{K_1}{N} \alpha_1 + \frac{K_2}{N} \alpha_2}.

    \end{array}
\end{equation}

\end{proof}

\begin{theorem}\label{theorem:bound2}

When there are two types in the environment, and they have been considered to be different types, the average deviation induced by the mean field approximation is bounded as follows,

\begin{equation}
    \frac{\sum_k || \hat{\delta}_k||_2}{N} \leq \frac{K_1}{N} \epsilon_1 + \frac{K_2}{N} \epsilon_2.
\end{equation}

The variables have the same meaning as in Theorem \ref{theorem:bound1}.

\end{theorem}
\begin{proof}

In this scenario we will have,

\begin{equation}\label{eq:multitypebound}
    \begin{array}{l}
  \medmath{       \frac{\sum_k || \hat{\delta}_k||_2}{N}  =  \frac{1}{N} (\sum_{k_1} ||a_{k_1} - \overline{a}_1 ||_2 + \sum_{k_2}||a_{k_2} - \overline{a}_2 ||_2)}
 \\ \\
 
 \medmath{ = \frac{K_1}{N} \frac{\sum_{k_1} ||a_{k_1} - \overline{a}_1 ||_2}{K_1} + \frac{K_2}{N}\frac{\sum_{k_2}||a_{k_2} - \overline{a}_2 ||_2}{K_2} \leq \frac{K_1}{N} \epsilon_1 + \frac{K_2}{N} \epsilon_2}.

    \end{array}{}
\end{equation}{}
\end{proof}

We presented Theorems \ref{theorem:bound1} and \ref{theorem:bound2} to demonstrate the reduction in the bound on the average deviation as we increase the number of types from 1 to 2.  Similar derivations can be performed to demonstrate that the bounds on the average deviation decrease as we increase the number of types (regardless of the true number of types).

Let $\epsilon$ be a bound on the average deviation achieved based on a certain number of types: $\frac{\sum_{k=1}^{X} ||\delta a||_2}{X} \leq \epsilon$.  The following theorem bounds the error of the approximate mean field $Q$ function as a function of $\epsilon$ and the smoothness $L$ of the exact $Q$ function.

\begin{theorem}\label{theorem:mtmfq}
When the $Q$ function is additively decomposable according to Equation~\ref{eq:meanfield}, and it is $L$-smooth, then the Multi Type Mean Field $Q$ function provides a good approximation bounded by 
\begin{equation}
    |Q^j(s,{\bf a}) - Q_{\textit{MTMF}}^j(s,a^j,\bar{a}^j_1,\ldots,\bar{a}^j_M)| \le \frac{1}{2}L\epsilon.
\end{equation}
\end{theorem}
\begin{proof}
We rewrite the expression for the $Q$ function as $Q(a) \triangleq Q^j(s,a^j,\overline{a}^j_1, \overline{a}^j_2, \cdots, \overline{a}^j_M)$. Suppose that $Q$ is $L$-smooth, where its gradient $\nabla Q$ is Lipschitz-continous with constant $L$ such that for all $a, \overline{a}$,

\begin{equation}
\begin{array}{l}
|| \nabla Q(a) - \nabla Q(\overline{a}) ||_2 \leq L ||a - \overline{a} ||_2
\end{array}
\end{equation}

\noindent where $||\cdot||_2$ denotes the $l_2$-norm.  Note that all the eigenvalues of $\nabla^2 Q$ can be bounded in the symmetric interval $[- L, L]$.
As the Hessian matrix $\nabla^2 Q$ is real symmetric and hence diagonalizable, there exists an orthogonal matrix $U$ such that $U^T[\nabla^2 Q]U = \Lambda  \triangleq diag[\lambda_1 ,\ldots, \lambda_D]$. It then follows that,

\begin{equation}\label{eq:bound}
\begin{array}{l}
 \delta a \cdot \nabla^2 Q \cdot \delta a = [U \delta a]^T \Lambda [U\delta a] = \sum_{i=1}^{D}\lambda_i[U \delta a]^2_i

 \end{array}
\end{equation}

with 

\begin{equation}
\begin{array}{l}

 -L||U\delta a ||_2 \leq \sum_{i = 1}^{D} \lambda_i[U \delta a]^2_i \leq L||U\delta a ||_2.

 \end{array}
\end{equation}

Let, $ \hat{\delta}_m a = a^j_m - \overline{a}^j_m $, consistent with the previous definition. Recall that the term $\delta$ is then $ \delta a = [\hat{\delta_1} a ; \hat{\delta_2} a ; \cdots ; \hat{\delta}_M a] $, where $a$ is the one-hot encoding for $D$ actions, and $\overline{a}$ is a $D$-dimensional categorical distribution. Then, it can be shown that,

\begin{equation}\label{eq:trunc}
\begin{array}{l}
||U(\delta a) ||_2 = || \delta a||_2. 
 \end{array}
\end{equation}

Consider the term $
\frac{1}{2 X} \sum_{k=1}^{X} R^j_{s,a^j}(a^k)
$. Since $L$ is the maximum eigenvalue, from Equation \ref{eq:bound} (with a slight abuse of notation),

\begin{equation}
R = \delta a \cdot \nabla^2 Q \cdot \delta a = [U\delta a]^T \Lambda [U\delta a]  = \sum_i \lambda_i[U \delta a]_i^2 \leq L||U \delta a||_2.
\end{equation}

Therefore, from Equation \ref{eq:trunc}:
\begin{equation}
 R  \leq  L||U \delta a||_2 = L||\delta a||_2.
 \end{equation}

Thus,

\begin{equation}
\medmath{\frac{1}{2 X} \sum_k R   \leq \frac{L}{2}\sum_k \frac{||\delta a||_2}{X}   
     \leq \frac{L\epsilon}{2}}. 
\end{equation}

\end{proof}
Thus, in this paper, we modify the mean field $Q$ function to include a finite number of types that each have a corresponding mean field. The $Q$ function then considers a finite number of interactions across types.

\subsection{Mean Field Update}

The mean action $\overline{a}^j_i$ represents the mean action of the neighbours of agent $j$ belonging to type $i$. 
As in the paper by Yang et al.~\shortcite{pmlr-v80-yang18d}, the mean field $Q$ function can be updated in a recurrent manner, 

\begin{equation}
\label{eq:MTMFQ}
\begin{array}{l}
\medmath{Q^j_{t + 1}(s,a^j, \overline{a}^j_1,\ldots,\overline{a}^j_M)  = 
(1-\alpha) Q_t^j(s,a^j, \overline{a}^j_1, \ldots, \overline{a}^j_M) + \alpha[r^j + \gamma v^j_t(s')] } 
\end{array}
\end{equation}

\noindent 
where $r^j$ is the reward obtained. The $s$ and $s'$ are the old and new states respectively. $\alpha_t$ is the learning rate. The value function $v^j_t(s')$ for agent $j$ at time $t$ is given by,

\begin{equation}\label{eq:valuefunc}
\begin{array}{l}
\medmath{v^{j}_{t}(s') = \sum_{a^j}\pi^j_t(a^j|s',\overline{a}^j_1,\ldots,\overline{a}^j_M)  \E_{a^{-j}\sim \pi_{t}^{-j}} [Q_t^j(s',a^j, \overline{a}^j_1, \ldots,\overline{a}^j_M)]}
\end{array}    
\end{equation}

\noindent Here, the term $\overline{a}_i^{j}$ denotes the mean action of all the other agents apart from $j$ belonging to type $i$.  In all of our implementations, the mean action for all the types is first calculated using the relation 

\begin{equation}\label{eq:newmean}
   \medmath{ \overline{a}^j_i = \frac{1}{N^j_i} \sum_{k} a^k_i, a^k_i \sim \pi_{t}^k(\cdot|s,\overline{a}^k_{1-},\ldots,\overline{a}^k_{M-})}
\end{equation}

\noindent
where $\pi_t^k$ is the policy of agent $k$ (in $j$'s neighbourhood) and $\overline{a}^k_{i-}$ represents the previous mean action for the neighbours of agent $k$ belonging to type $i$. $N^j_i$ is the total number of agents of type $i$ in $j$'s neighbourhood. Then, the Boltzmann policy for each agent $j$ is
\begin{equation}\label{eq:boltz}
\medmath{\pi_t^j(a^j|s, \overline{a}^j_1,\ldots,\overline{a}^j_M) = \frac{exp(\beta Q^j_t(s,a^j, \overline{a}^j_1,\ldots,\overline{a}^j_M))}{\sum_{a^{j'}\in A^j}exp(\beta Q_t^j(s,a^{j'},\overline{a}^{j}_1,\ldots,\overline{a}^j_M))} }
\end{equation}

\noindent 
where $\beta$ is the Boltzmann softmax parameter. 

By iterating through Equations \ref{eq:valuefunc}, \ref{eq:newmean} and \ref{eq:boltz}, the mean actions and respective policies of all agents keep improving. We prove in Theorem~\ref{theorem:nasheq} that this approach converges to a fixed point within a small bounded distance of the Nash equilibrium.  In Appendix A,
we give a specific example for a case in which the Multi Type Mean Field algorithm does better than the simple mean field method.

We first make three mild assumptions about the Multi Type Mean Field update and then state two lemmas before giving Theorem~\ref{theorem:nasheq}. 

\textbf{Assumption 1}: In the Multi Type Mean Field update, each action-value pair is visited infinitely often, and the reward is bounded by some constant.

\textbf{Assumption 2}: The agent's policies are Greedy in the Limit with Infinite Exploration (GLIE). In the case of the Boltzmann policy, the policy becomes greedy w.r.t. the $Q$ function in the limit as the temperature decays asymptotically to zero.

\textbf{Assumption 3}: For each stage game $[Q^1_t(s),\ldots,Q^N_t(s)]$ at time~$t$ and in state $s$ in training, for all $t,s,j \in \{1,\ldots,N\}$ the Nash equilibrium $\pi_{*} = [\pi_{*}^1,\ldots,\pi_{*}^N]$ is recognized either as a global optimum or a saddle point as expressed as:
\begin{equation}\label{eq:gop}
\begin{array}{l}

\medmath{1. \E_{\pi_*}[Q_t^j(s)] \geq \E_{\pi}[Q_t^j(s)], \forall \pi \in \Omega(\Pi_k \mathscr{A}^k)}.

\end{array}
\end{equation}

\begin{equation}\label{eq:sp1}
\begin{array}{l}

\medmath{2. \E_{\pi_*}[Q_t^j(s)] \geq \E_{\pi^j}\E_{\pi_*^{-j}}[Q_t^j(s)], \forall \pi^j \in \Omega(\mathscr{A}^j)}, \textrm{and}

\end{array}
\end{equation}

\begin{equation}\label{eq:sp2}
\begin{array}{l}
\medmath{\E_{\pi_*}[Q_t^j(s)] \leq \E_{\pi^j_*}\E_{\pi^{-j}}[Q_t^j(s)], \forall \pi^{-j} \in \Omega(\Pi_{k \neq j} {\mathscr{A}^k})}.
\end{array}
\end{equation}

\begin{lemm}\label{lemma:nashoperator}

Under Assumption 3, the Nash operator $\mathscr{H}^{Nash}$ forms a contraction mapping under the complete metric space from $\mathcal{Q}$ to $\mathcal{Q}$ with the fixed point being the Nash $Q$ value of the entire game ($\boldsymbol{Q}_*$), i.e., $ \mathscr{H}^{Nash} \boldsymbol{Q}_* = \boldsymbol{Q}_*  $.
\end{lemm}{}
\begin{proof}
Refer to Theorem 17 in \cite{hu2003nash} for a detailed proof. 
\end{proof}{}

We define a new operator $\mathscr{H}^{MTMF}$ which is the Multi Type Mean Field operator. We differentiate this from the Nash operator used above. This operator is defined as (where $\boldsymbol{Q} \triangleq [Q^1, \ldots, Q^N]$),

\begin{equation}
    \begin{array}{l}
        \mathscr{H}^{MTMF} \boldsymbol{Q}(s,\textbf{a}) \triangleq \E_{s' \sim p}[ \textbf{r}(s,\textbf{a}) + \gamma \textbf{v}^{MTMF} (s')].
    \end{array}
\end{equation}

The Multi Type Mean Field value function can be defined as $\textbf{v}^{MTMF}(s) \triangleq [v^1(s), \ldots, v^N(s)]$. This is the value function obtained from Equation \ref{eq:valuefunc}. Also, $\boldsymbol{r}(s,\boldsymbol{a})  = [r^1(s,\boldsymbol{a}),\ldots,r^N(s,\boldsymbol{a})]$. Now using the same principle of Lemma \ref{lemma:nashoperator} on the Multi Type Mean Field operator, we can show that the Multi Type Mean Field operator also forms a contraction mapping (additionally refer to Proposition~1 in Yang et al.~\cite{pmlr-v80-yang18d}).

\begin{lemm}\label{lemma:jakkollalemma}

The random process ${\Delta_t}$ defined in $\mathcal{R}$ as

\begin{equation}\label{eq:lemma2}
\begin{array}{l}
\Delta_{t+1}(x) = (1 - \alpha)\Delta_t(x) + \alpha F_t(x)  

\end{array}
\end{equation}
converges to a constant $S$ with probability 1 (w.p.t 1) when 
\begin{equation}
\begin{array}{l}
1) \hspace{10mm} 0 \leq \alpha \leq 1, \sum_t \alpha = \infty, \sum_t \alpha^2  < \infty.
\end{array}
\end{equation}

\begin{equation}
\begin{array}{l}
2) \hspace{30mm} x \in \mathscr{X}; |\mathscr{X}| <  \infty
\end{array}
\end{equation}

\noindent where $\mathscr{X}$ is the set of possible states,

\begin{equation}
\begin{array}{l}
3) \hspace{10mm} ||\E[F_t(x)|\mathscr{F}_t]||_w \leq \gamma || \Delta_t||_w + K   
\end{array}
\end{equation}

\noindent where $\gamma \in [0,1)$ and K is finite. 

\begin{equation}
\begin{array}{l}
4) \hspace{10mm}  \textbf{var}[F_t(x)|\mathscr{F}_t] \leq K_2(1 + || \Delta_t||^2_w)   
\end{array}
\end{equation}
with constant $K_2 > 0$ and finite. 

Here $\mathscr{F}_t$ denotes the filtration of an increasing sequence of $\sigma$-fields including the history of processes; $\alpha_t$, $\Delta_t$, $F_t \in \mathscr{F}_t$ and $|| \cdot ||_w $ is a weighted maximum norm. The value of this constant $S = \frac{\psi C_1 + \alpha |K|}{\alpha \beta_0}$, where $\psi \in (0,1)$ and $C_1$ is the value with which the scale invariant iterative process is bounded. $\beta_0$ is the scale factor applied to the original process.
\end{lemm}{}

\begin{proof}

This lemma follows from Theorem 1 in \cite{jaakkola1994convergence}. We provide the complete proof of this lemma in the Appendix C, 
highlighting the changes from the Theorem 1 in \cite{jaakkola1994convergence}. 

\end{proof}{}

\begin{theorem}\label{theorem:nasheq}
When updating $Q^j(s,a^j,\bar{a}^j_1,\ldots,\bar{a}^j_M)$ according to Equations~\ref{eq:MTMFQ}, \ref{eq:valuefunc}, \ref{eq:newmean} and \ref{eq:boltz}, for all agents $j$, the multi-agent $Q$ function will converge to a bounded distance of the Nash $Q$ function under the Assumptions~1, 2 and 3, expressed as,

\begin{equation*}
\begin{array}{l}
      \textbf{Q}_{*}(s,\boldsymbol{a}) - \textbf{Q}_t(s, \boldsymbol{a} )  \leq D - S.
\end{array}{}
\end{equation*}
\noindent
where $S = \frac{\psi C_1 + \alpha \gamma |D|}{\alpha \beta_0} $. Here $D = \frac{1}{2}L\epsilon$, from Theorem \ref{theorem:mtmfq}. The joint Nash $Q$ function is denoted as $\boldsymbol{Q}_* = [Q^1_*, \ldots, Q^N_*]$, where $Q^j_*$ denotes the Nash Q-value of the agent $j$ (value received by the agent $j$ in a Nash equilibrium), and $\boldsymbol{Q}_t = [Q^1_t, \ldots, Q^N_t] $.

\end{theorem}

\begin{proof}

The proof of convergence of this theorem is structurally similar to that discussed by the authors in \cite{hu2003nash} and \cite{pmlr-v80-yang18d}. We provide the proof here, with changes necessitated by the multiple type case.

Note that in Assumption 3, Equation \ref{eq:gop} corresponds to the global optimum and Equations \ref{eq:sp1} and \ref{eq:sp2} correspond to the saddle point. These assumptions are the same as those considered in \cite{hu2003nash}. Also note that the authors in \cite{hu2003nash} and \cite{pmlr-v80-yang18d} mention that Assumption 3 is a strong assumption to impose, which is needed to show the theoretical convergence, but this is not required to impose in practice. 

From \cite{hu2003nash}, we formally define the Nash operator $\mathscr{H}^{Nash}$ as,

\begin{equation}
\begin{array}{l}
\medmath{ \mathscr{H}^{Nash} \boldsymbol{Q}(s,\boldsymbol{a}) \triangleq \E_{s' \sim p} [\boldsymbol{r}(s,\boldsymbol{a}) + \gamma \textbf{v}^{Nash}(s')]}.
\end{array}
\end{equation}

\noindent
where $\boldsymbol{Q} \triangleq [Q^1,\ldots,Q^N]$ and $\boldsymbol{r}(s,\boldsymbol{a})  \triangleq [r^1(s,\boldsymbol{a}),\ldots,r^N(s,\boldsymbol{a})]$.

\noindent The Nash value function is $\textbf{v}^{Nash}(s) \triangleq  [v^1_{\boldsymbol{\pi_*}}(s), \ldots, v^N_{\boldsymbol{\pi_*}}(s)]
$. Here the joint Nash policy is represented as $\boldsymbol{\pi_*}$. The Nash value function is calculated with the assumption that all agents are following $\boldsymbol{\pi_*}$ from the initial state $s$. 

We need to apply Lemma \ref{lemma:jakkollalemma} to prove Theorem \ref{theorem:nasheq}. 
By subtracting $Q_{*}(s,\boldsymbol{a})$ on both sides of Equation \ref{eq:MTMFQ} and in relation to Equation~\ref{eq:lemma2}, 

\begin{equation}\label{eq:deltaandF}
\begin{array}{l}
\medmath{\Delta_t(x) = \textbf{Q}_t(s,a^j, \overline{a}^j_1, \ldots, \overline{a}^j_M ) - \textbf{Q}_{*}(s,\boldsymbol{a})}\\

\medmath{\textbf{F}_t(x) = \textbf{r}_t + \gamma \textbf{v}_t^{MTMF}(s_{t+1}) - \textbf{Q}_{*}(s_t, \boldsymbol{a}_t)}

\end{array}
\end{equation}

\noindent where $x \triangleq (s_t,\boldsymbol{a}_t)$ denotes the visited state-joint action pair at the time $t$. 

In Theorem \ref{theorem:mtmfq}, we proved a bound for the actual $Q$ function and the Multi Type Mean Field $Q$ function. We apply that to Equation~\ref{eq:deltaandF}, to get the following equation for $\Delta$,

\begin{equation}\label{eq:changeddelta}
    \begin{array}{l}
        
    \medmath{\Delta_t(x) = \textbf{Q}_t(s,a^j, \overline{a}^j_1, \ldots, \overline{a}^j_M ) - \textbf{Q}_{*}(s,\boldsymbol{a})}
    
    \\ 
    
    \medmath{\Delta_t(x) = \textbf{Q}_t(s,a^j, \overline{a}^j_1, \ldots, \overline{a}^j_M ) + \textbf{Q}_t(s, \boldsymbol{a} ) - \textbf{Q}_t(s, \boldsymbol{a} ) -  \textbf{Q}_{*}(s,\boldsymbol{a})}
    
       \\ 
    
    \medmath{\Delta_t(x) \leq |\textbf{Q}_t(s,a^j, \overline{a}^j_1, \ldots, \overline{a}^j_M ) - \textbf{Q}_t(s, \boldsymbol{a} )| + \textbf{Q}_t(s, \boldsymbol{a} ) -  \textbf{Q}_{*}(s,\boldsymbol{a})}

    \\
    
        \medmath{ \Delta_t(x) \leq \textbf{Q}_t(s, \boldsymbol{a} ) - \textbf{Q}_{*}(s,\boldsymbol{a}) + D}
    \end{array}
\end{equation}

\noindent where $D = \frac{1}{2} L\epsilon$.

The aim is to prove that the four conditions of Lemma \ref{lemma:jakkollalemma} hold and that $\Delta$ in Equation \ref{eq:changeddelta} converges to a constant $S$ according to Lemma \ref{lemma:jakkollalemma} and thus the MTMF $Q$ function in Equation \ref{eq:changeddelta} converges to a point whose distance to the Nash Equilibrium is bounded. In Equation~\ref{eq:lemma2}, $\alpha (t)$ refers to the learning rate and hence the first condition of Lemma \ref{lemma:jakkollalemma} is automatically satisfied. The second condition is also true as we are dealing with finite state and action spaces.

Let $\mathscr{F}_t$ be the $\sigma$-field generated by all random variables in the history time $t$ - $ (s_t, \alpha_t, a_t, r_{t-1}, \ldots ,s_1,\alpha_1, \textbf{a}_1,\textbf{Q}_0) $. 
Thus, $\textbf{Q}_t$ is a random variable derived from the historical trajectory up to time $t$.

To prove the third condition of Lemma \ref{lemma:jakkollalemma}, from Equation \ref{eq:deltaandF},
\begin{equation}\label{eq:Freduction}
\begin{array}{l}

\medmath{\boldsymbol{F}_t(s_t, \boldsymbol{a}_t) = \textbf{r}_t + \gamma \textbf{v}_t^{MTMF} - \textbf{Q}_{*}(s_t,\boldsymbol{a}_t)  }
\\
\medmath{ = \textbf{r}_t + \gamma \textbf{v}_t^{Nash} (s_{t+1}) - \textbf{Q}_{*}(s_t,\boldsymbol{a}_t) + \gamma[\textbf{v}_t^{MTMF}(s_{t+1}) - \textbf{v}_t^{Nash}(s_{t+1})] }\\
 
\medmath{= (\textbf{r}_t + \gamma \textbf{v}_t^{Nash} (s_{t+1}) - Q_{*}(s_t,\boldsymbol{a}_t)) + C_t(s_t, \boldsymbol{a}_t)
 \triangleq \boldsymbol{F}_t^{Nash}(s_t, \boldsymbol{a}_t) + C_t(s_t,\boldsymbol{a}_t)}.

\end{array}
\end{equation}

From Lemma \ref{lemma:nashoperator}, $\boldsymbol{F}_t^{Nash}$ forms a contraction mapping with the norm $|| \cdot ||_{\infty}$ being the maximum norm on $\boldsymbol{a}$. Thus, from Equation~\ref{eq:changeddelta} we get, 

\begin{equation}\label{eq:lemma1appl}
\begin{array}{l}

\medmath{ ||\E[\boldsymbol{F}_t^{Nash}(s_t,\boldsymbol{a}_t)|\mathscr{F}_t]||_{\infty} \leq \gamma || \boldsymbol{Q}_* - \boldsymbol{Q}_t||_\infty \leq \gamma || D - \Delta_t||_\infty}.

\end{array}
\end{equation}

Now, applying Equation \ref{eq:lemma1appl} in Equation \ref{eq:Freduction}, 

\begin{equation}\label{eq:lemma2proof}
\begin{array}{l}
 \medmath{|| \E[F_t(s_t, \boldsymbol{a}_t)|\mathscr{F}_t]||_\infty   =  || F_t^{Nash}(s_t,\boldsymbol{a}_t)|\mathscr{F}_t||_\infty + || \boldsymbol{C}_t(s_t,\boldsymbol{a}_t)|\boldsymbol{\mathscr{F}}_t||_\infty}  \\
\medmath{\leq \gamma || D - \Delta_t||_\infty + || \boldsymbol{C}_t(s_t,\boldsymbol{a}_t)|\boldsymbol{\mathscr{F}}_t||_\infty}  \\
\medmath{\leq \gamma || \Delta_t||_\infty + || \boldsymbol{C}_t(s_t,\boldsymbol{a}_t)|\boldsymbol{\mathscr{F}}_t||_\infty  + \gamma || D ||_\infty \leq \gamma || \Delta_t||_\infty + \gamma |D|}.

\end{array}
\end{equation}

Since we are taking the max norm, the last two terms in the right-hand side of Equation \ref{eq:lemma2proof} are both positive and finite. We can prove that the term $||C_t(s_t, \boldsymbol{a}_t)||$ converges to 0 w.p.1. The proof involves the use of Assumption 3 (refer to Theorem 1 in \cite{pmlr-v80-yang18d}). We use this fact in the last step of Equation \ref{eq:lemma2proof}.
Hence, the third condition of Lemma \ref{lemma:jakkollalemma} is satisfied. The value of constant $K = \gamma |D| = \gamma |\frac{1}{2}L\epsilon|$.

For the fourth condition we use the fact that the MTMF operator $\mathscr{H}^{MTMF}$  forms a contraction mapping. Hence, $\mathscr{H}^{MTMF}\textbf{Q}_* = \textbf{Q}_*$ and it follows that, 
\begin{equation}
\begin{array}{l}
\medmath{ \textbf{var}[\boldsymbol{F}_t(s_t, \boldsymbol{a}_t)|\mathscr{F}_t] = E[(r_t + \gamma \textbf{v}_t^{MTMF}(s_{t+1}) - \boldsymbol{Q}_*(s_t,\boldsymbol{a}_t))^2] }
\\
 \medmath{ = E[(\boldsymbol{r}_t + \gamma \textbf{v}_t^{MTMF}(s_{t+1}) - \mathscr{H}^{MTMF}(\boldsymbol{Q}_*))^2]} \\ 
 
 \medmath{= \textbf{var}[\boldsymbol{r}_t + \gamma \textbf{v}_t^{MTMF}(s_{t+1})|\mathscr{F}_t] \leq K_2(1 + ||\Delta_t||^2_{W})}.
 
 \end{array}
 \end{equation}

In the last step, the left side of the equation contains the reward and the value function as the variables. The reward is bounded by Assumption 1 and the value function is also bounded by being updated recursively by Equation \ref{eq:valuefunc} (MTMF is a contraction operator). So we can choose a positive, finite $K_2$ such that the inequality holds.

Finally, with all conditions met, it follows from Lemma \ref{lemma:jakkollalemma} that $\Delta_t$ converges to constant $S$ w.p.1. The value of this constant is $S = \frac{\psi C_1 + \alpha \gamma |D|}{\alpha \beta_0}$ from Lemma 2 and using the value of $K$ derived above. Therefore, from Equation \ref{eq:changeddelta} we get,

\begin{equation}
    \begin{array}{l}

    \medmath{ \textbf{Q}_{*}(s,\boldsymbol{a}) - \textbf{Q}_t(s, \boldsymbol{a} )  \leq D - S \leq \frac{1}{2}L\epsilon - S}.

    \end{array}
\end{equation}

Hence, the multi-agent $Q$ function converges to a point within a bounded distance from the Nash equilibrium of the stochastic game.  The distance is a function of the error in the type classification and the closeness of resemblance of each agent to its type. 
\end{proof}

\section{Implementation}

We propose two algorithms based on Q-learning to estimate the Multi Type Mean Field $Q$ function for known and unknown types. The algorithms, denoted as MTMFQ (Multi Type Mean Field Q-learning), train an agent $j$ to minimize the loss function $\mathcal{L} (\phi^j) = (y^j - Q_{\phi^j}(s, a^j, \overline{a}^j_1,\dots,\overline{a}^j_M))^2 $. Here $y^j = r^j + \gamma v^{\textit{MTMF}}_{\phi^j_{\_}}(s')$ (from Eq.~\ref{eq:MTMFQ}) is the target value used to calculate the temporal difference (T.D.) error using the weights $\phi^j_{\_}$. Now, the gradient is obtained as,

\begin{equation}
\begin{array}{l}
\medmath{\nabla_{\phi^j}\mathcal{L}(\phi^j) = 2(Q_{\phi^j}(s, a^j, \overline{a}^j_1,\dots,\overline{a}^j_M) - y^j)  \times 
\nabla_{\phi^j} Q_{\phi^j} (s, a^j, \overline{a}^j_1,\dots,\overline{a}^j_M).}
\end{array}
\end{equation}

Algorithm \ref{alg:MTMFQ} describes the Multi Type Mean Field Q-learning (MTMFQ) algorithm when agent types are known. In this algorithm, different groups of agents in the environment are considered as types. An agent models its relation to each type separately and ultimately chooses the action that provides maximum benefit in the face of competition against the different types. This is referred to as version 1 of MTMFQ. This algorithm deals with multiple types in contrast to MFQ described in the paper by Yang et al.~\shortcite{pmlr-v80-yang18d}. In Line 8, each agent is chosen, and its neighbours are considered. The neighbours are classified into different types and in each type a new mean action is calculated (Line 9). In Lines 14 -- 19, the Q networks are updated as done in common practice \cite{mnih2015human} for all the agents in the environment. At Line 12, the current actions are added to a buffer containing previous mean actions. 

Version 2 of MTMFQ (see Algorithm \ref{alg:MTMFQ-2}) deals with the second type of scenario, where the type of each agent is unknown. 
An additional step doing K-means clustering is introduced to determine the types. Once we recognize the type of each agent, the algorithm is very similar to Algorithm \ref{alg:MTMFQ}.  The clustering does not necessarily recognize the types correctly and the overall process is not as efficient as the known type case. But we will show that this way of approximate type determination is still better than using only a single mean field for all the agents. For the implementation we use neural networks, but it can be done without neural networks too. We only need a way to recursively update Equations \ref{eq:valuefunc}, \ref{eq:newmean} and \ref{eq:boltz}. 

Note that the total number of types in the environment is unknown, and the agent does not need to guess the correct number of types.  Generally, when more types are used, the approximate multi-agent $Q$ function will be closer to the exact Nash $Q$ function as shown by the bounds in Theorems~\ref{theorem:bound1}, \ref{theorem:bound2}, \ref{theorem:mtmfq} and \ref{theorem:nasheq}.  There is no risk of overfitting when using more types.  In the limit, when there is one type per agent, we recover the exact multi-agent $Q$ function. The only drawback is an increase in computational complexity. The code repository will appear at  \url{https://github.com/BorealisAI/mtmfrl}. 

\begin{algorithm}
\caption{Multi Type Mean Field Q-learning for known types}
\label{alg:MTMFQ}
\begin{algorithmic}[1] %[1] enables line numbers

\STATE Initialize the number of types $M$ and total number of agents $N$.
\STATE Initialize the $Q$ functions (parameterized by weights) $ Q_{\phi^j}, Q_{\phi^j_{\_}}$, for all agents $j$. 
\STATE Initialize the mean action for each type $\overline{a}^j_1$, $\overline{a}^j_2, \ldots, \overline{a}^j_M$, for each agent $j \in {1, \ldots, N}$.
\STATE Initialize the total number of steps (T) and total number of episodes (E).
\WHILE {Episode $<$ E}
\WHILE{Step $<$ T}
\STATE For each agent $j$ choose action $a^j$ from $Q_{\phi^j}$ according to Eq. \ref{eq:boltz} with the current mean action for each type $\overline{a}^{j}_1, \ldots, \overline{a}^{j}_M $ and the exploration rate $\beta$.
\STATE For each agent j, compute the new mean action for each type  $\overline{a}^{j}_1, \ldots, \overline{a}^{j}_M$ according to Eq. \ref{eq:newmean}.
\STATE Execute the joint action $\textbf{a} = [a^1,\ldots,a^N]$. Observe the rewards $\textbf{r} = [r^1,\ldots,r^N]$ and the next state $s'$. 
\STATE Store $\langle s,\textbf{a},\textbf{r},s',\overline{\textbf{a}_1}, \ldots, \overline{\textbf{a}_M} \rangle$ in replay buffer $D$, where $\overline{\textbf{a}_i}$ is the mean action for type $i$ in the neighbourhood. The $\textbf{a}$ captures all the $N$ agents. 
\ENDWHILE
\WHILE{$j$ = $1$ to $N$}
\STATE Sample a minibatch of $K$ experiences  $\langle s,\textbf{a},\textbf{r},s',\overline{\textbf{a}_1}, \ldots, \overline{\textbf{a}_M} \rangle$
from $D$.
\STATE Sample action $a^j_{\_}$ from $Q_{\phi^j_{\_}}$ with $\overline{a}^{j}_{i\_}  \leftarrow {\overline{a}^{j}_{i}}$ for each type $i$. 
\STATE Set $y^j = r^j + \gamma v^{MTMF}_{\phi^j_{\_}}(s')$ according to Eq. \ref{eq:MTMFQ}.
\STATE Update the Q network by minimizing the loss $L(\phi^j) = \frac{1}{K} \sum (y^j - Q_{\phi^j}(s^j, a^j, \overline{a}^{j}_1,
\ldots, \overline{a}^j_M))^2$.
\ENDWHILE
\STATE Update the parameters of the target network for each agent $j$ with learning rate $\tau$; $\phi^{j}_{\_} \leftarrow \tau \phi^j + (1 - \tau) \phi^j_{\_}$.
\ENDWHILE
\end{algorithmic}
\end{algorithm}

\begin{algorithm}
\caption{Multi Type Mean Field Q-learning for unknown types}
\label{alg:MTMFQ-2}
\begin{algorithmic}[1] %[1] enables line numbers

\STATE Initialize the number of types $M$ and total number of agents $N$.
\STATE Initialize the $Q$ functions (parameterized by weights) $ Q_{\phi^j}, Q_{\phi^j_{\_}}$, for all agents $j$. 
\STATE Initialize the mean action for each type $\overline{a}^j_1$, $\overline{a}^j_2, \ldots, \overline{a}^j_M$, for each agent $j \in {1, \ldots, N}$.
\STATE Initialize the total number of steps (T) and total number of episodes (E).
\STATE Initialize every agent to a type at random. Initialize an array $A$ containing the previous action of all agents.  
\STATE Maintain a buffer $B$ for storing the last $C$ actions of all agents. $C$ is determined by the conditions of the environment. Initialize all values to 0. 
\WHILE {Episode $<$ E}
\WHILE{steps $<$ T}
\STATE For each agent $j$ choose action $a^j$ from $Q_{\phi^j}$ according to Eq.~\ref{eq:boltz} with the current mean action for each type $\overline{a}^{j}_1, \ldots, \overline{a}^{j}_M$ and the exploration rate $\beta$.
\STATE For each agent j, compute the new mean action for each type  $\overline{a}^{j}_1, \ldots, \overline{a}^{j}_M$ according to Eq. \ref{eq:newmean}.
\STATE Execute the joint action $\textbf{a} = [a^1,\ldots,a^N]$. Observe the rewards $\textbf{r} = [r^1,\ldots,r^N]$ and the next state $s'$. 
\STATE Store $\langle s,\textbf{a},\textbf{r},s',\overline{\textbf{a}_1}, \ldots, \overline{\textbf{a}_M} \rangle$ in replay buffer $D$, where $\overline{\textbf{a}_i}$ is the mean action for type $i$ in the neighbourhood. The $\textbf{a}$ captures all the $N$ agents. 
 
\STATE Store each action $[a^1,\ldots,a^N]$ in the array $A$. Update the Buffer $B$ with the last action taken by all agents.  
\STATE Perform a K-means clustering on $B$ with  the number of clusters equal to the number of types $M$. 
\STATE Reassign the agents to different types based on the cluster in K-means. 
\ENDWHILE
\WHILE{$j$ = $1$ to $N$}
\STATE Sample a minibatch of $K$ experiences  $\langle s,\textbf{a},\textbf{r},s',\overline{\textbf{a}_1}, \ldots, \overline{\textbf{a}_M} \rangle$
from $D$.
\STATE Sample action $a^j_{\_}$ from $Q_{\phi^j_{\_}}$ with $\overline{a}^{j}_{i\_}  \leftarrow {\overline{a}^{j}_{i}}$ for each type $i$. 
\STATE Set $y^j = r^j + \gamma v^{MTMF}_{\phi^j_{\_}}(s')$ according to Eq. \ref{eq:MTMFQ}.
\STATE Update the Q network by minimizing the loss $L(\phi^j) = \frac{1}{K} \sum (y^j - Q_{\phi^j}(s^j, a^j, \overline{a}^{j}_1,\ldots,\overline{a}^j_M))^2$.
\ENDWHILE
\STATE Update the parameters of the target network for each agent $j$ with learning rate $\tau$; $\phi^{j}_{\_} \leftarrow \tau \phi^j + (1 - \tau) \phi^j_{\_}$.
\ENDWHILE
\end{algorithmic}
\end{algorithm}

\section{Experiments and Results}

\begin{figure}
	\subfloat[Game Domain Multi Team Battle]{{\includegraphics[width=4cm, height=3cm]{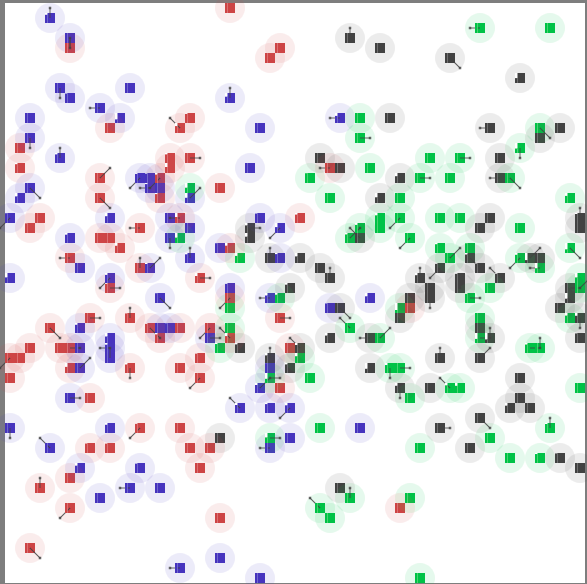} }}
	%\quad
	\subfloat[Multi Team Battle Training]{{\includegraphics[width=4cm, height=3cm]{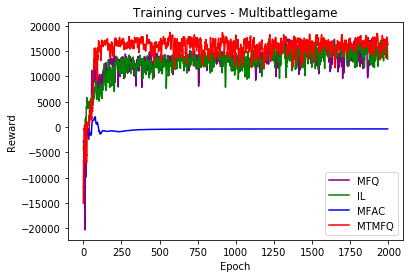} }}\\
	%\quad
	\subfloat[Win rate for each algorithm]{{\includegraphics[width=4cm, height=3cm]{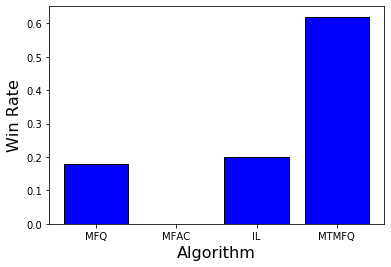} }}
	\subfloat[Total rewards taken as an average per episode]{{\includegraphics[width=4cm, height=3cm]{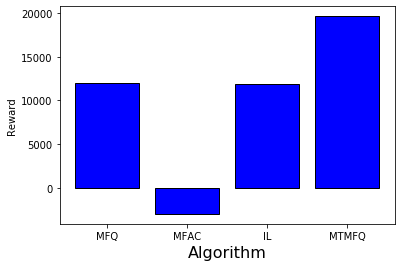} }}
  \caption{Multi Team Battle Game Results.}%
	\label{fig:multibattle}
\end{figure}

\begin{figure}
	\subfloat[Game Domain Battle-Gathering]{{\includegraphics[width=4cm, height=3cm]{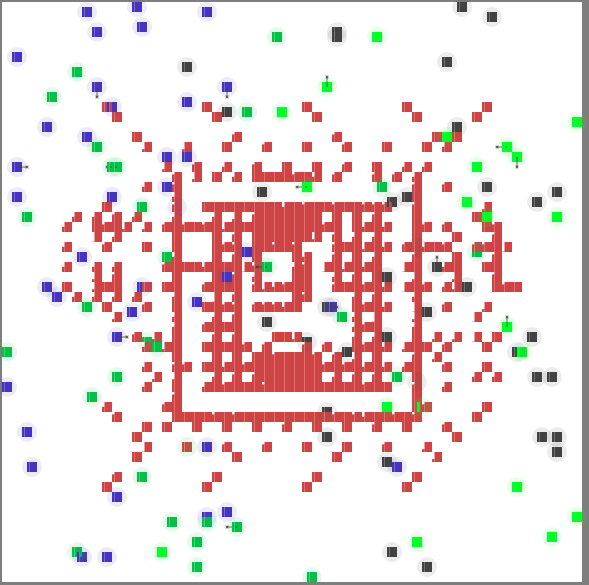} }}
	%\quad
	\subfloat[Battle-Gathering training]{{\includegraphics[width=4cm, height=3cm]{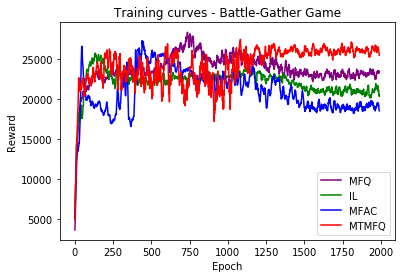} }}\\
	%\quad
	\subfloat[Win rate for each algorithm]{{\includegraphics[width=4cm, height=3cm]{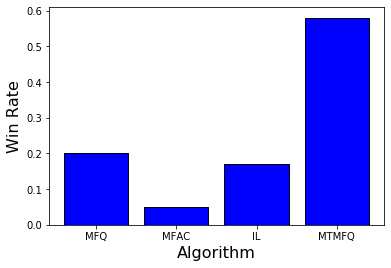} }}
	\subfloat[Total Rewards taken as an average per episode.]{{\includegraphics[width=4cm, height=3cm]{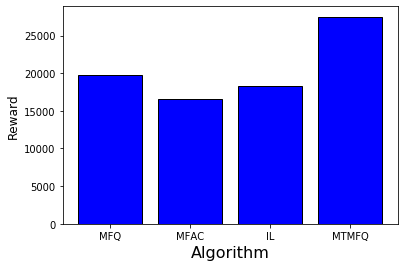} }}
  \caption{Battle-Gathering Game Results.}%
	\label{fig:gathering}
\end{figure}

\begin{figure}
	\subfloat[Game Domain Predator Prey]{{\includegraphics[width=4cm, height=3cm]{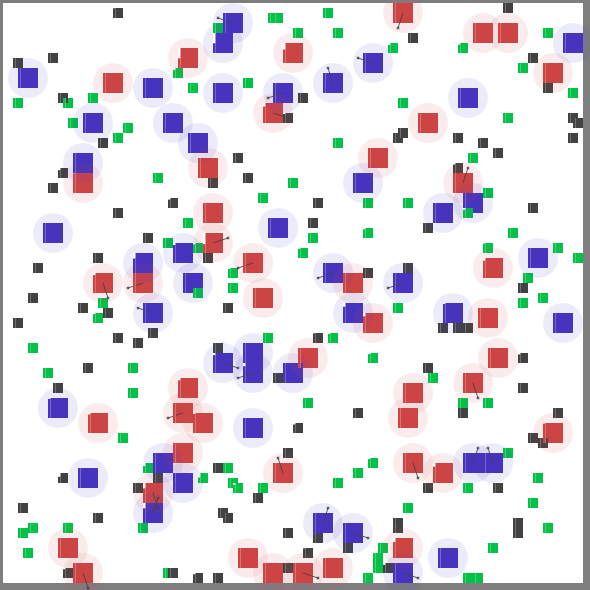} }}
	%\quad
	\subfloat[Predator Prey Training]{{\includegraphics[width=4cm, height=3cm]{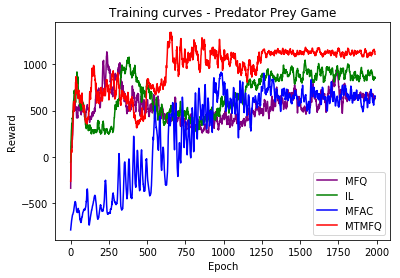} }}\\
	%\quad
	\subfloat[Win rate for each algorithm]{{\includegraphics[width=4cm, height=3cm]{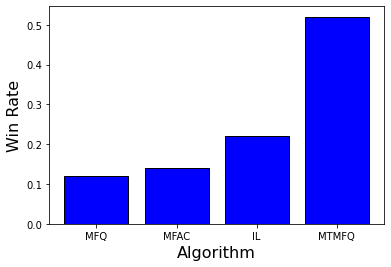} }}
	\subfloat[Total rewards taken as an average per episode]{{\includegraphics[width=4cm, height=3cm]{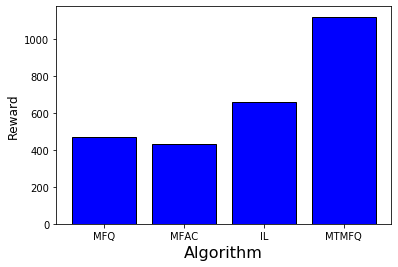} }}
  \caption{Predator Prey Game Results.}%
	\label{fig:predatorprey}
\end{figure}

We report results with three games designed within the MAgents framework: the Multi Team Battle, Battle-Gathering and the Predator Prey domains. 
In the first two games, Multi Team Battle and the Battle-Gathering game, the conditions are such that the different groups (or teams) are fully known upfront. In the third game, the conditions are such that the types of the agents are initially unknown. Hence, the agents must also learn the identity of the opponent agents during game play. Multi Team Battle and Gathering are two separately existing MAgent games, which have been combined to obtain the Battle-Gathering game used in this paper. Predator Prey domain is also obtained from combining Multi Team Battle with another existing MAgent game (Pursuit). In preparation for each game, agents train for 2000 episodes of game play against different groups training using the same algorithm, which is referred to as the first stage.  Next, in the second stage, they enter into a faceoff against other agents trained by other algorithms where they fight each other for 1000 games. We report the results for both stages. We repeat all experiments 50 times and report the averages. As can be seen from the nature of the experiments, variances can be quite large across individual experimental runs. 

In the first game (Multi Team Battle, see Figure~\ref{fig:multibattle}(a)), there are four teams (different colours in Figure~\ref{fig:multibattle}(a)) fighting against each other to win a battle. Here, agents belonging to one team are competing against agents from other teams and cooperating with agents of the same team. During the first stage, a team does not know how the other teams will play during the faceoff --- so it clones itself to three additional teams with slightly different rewards in order to induce different strategies. This is similar to self play. Each team receives a reward equal to the sum of the local rewards attributed to each agent in the team. The reward function is defined in such a way that it encourages local cooperation in killing opponents and discourages getting attacked and dying in the game. The reward function for different agent groups are also subtly different (refer to Appendix B for the details). Let us call each of these groups: Group A, Group B, Group C, and Group D. We maintain a notion of favourable opponents and each group gets slightly higher rewards for killing the favourable opponents than the others. Four algorithms, namely MFQ \cite{pmlr-v80-yang18d}, MFAC \cite{pmlr-v80-yang18d}, Independent Q Learning (IL) \cite{tan1993multi}, and MTMFQ are trained separately where each of these groups trains its own separate network using the same algorithm (all agents within a group train the same network). We start all the battles with 72 agents of each group for training and testing. Group A from MTMFQ, Group B from IL, Group C from MFQ, and Group D from MFAC enter the faceoff stage where they fight each other for 1000 games. Each game (or episode) has a maximum of 5000 steps and a game is considered to be won by the group/groups that have the most number of agents alive at the end of the game. 

The results of the first training stage are reported in Figure~\ref{fig:multibattle}(a).  We report the cumulative rewards from all agents in Group A for each algorithm.  These rewards are for Group A against the other 3 groups. Since the different groups are just clones of each other the reward curves for other groups are similar to that of Group A. In the training stage, the teams trained by different algorithms did not play against each other, but simply against the cloned teams trained by the same algorithm.  At the beginning of training, for about 100 episodes the agents are still exploring, and their strategies are not well differentiated yet. As a result, MTMFQ's performance is still comparable to the performance of other algorithms. At this stage, the assumption of a single type is fine. As training progresses, each group begins to identify the favourable opponents and tries to make a targeted approach in the battle. When such differences exist across a wide range of agents, the MTMFQ algorithm shows a better performance than the other techniques as it explicitly considers the presence of different types in the game. Overall, we observe that MTMFQ has a faster convergence than all other algorithms, and it also produces higher rewards at the end of the complete training. This shows that MTMFQ identifies favourable opponents early, but the other algorithms struggle longer to learn this condition. The MFAC algorithm is the worst overall. This is consistent with the observation by Yang et al.~\shortcite{pmlr-v80-yang18d}, where the authors give particular reasons for this bad performance, including the greedy in the limit with infinite exploration (GLIE) assumption and a positive bias of Q-learning algorithms. This algorithm is not able to earn an overall positive reward upon the complete training. Figure~\ref{fig:multibattle}(c) shows the win rate of the teams trained by each algorithm (MTMFQ, MFQ, MFAC, and IL) in a direct faceoff of 1000 episodes. In the face off, each group is trained by a different algorithm and with a different reward function that induces different strategies.  The only algorithm that handles different strategies among the opponent teams is MTMFQ, and therefore it dominates the other algorithms with a win rate of 60\%. Figure~\ref{fig:multibattle}(d) reinforces this domination.

The second domain is the  Battle-Gathering game (Figure~\ref{fig:gathering}(a)). In this game, all the agent groups compete for food resources that are limited (red denotes food particles and other colours are competing agents) in addition to killing its opponents as in the Multi Team Battle game. Hence, this game is harder than the first one. All the training and competition are similar to the first game.

Figure~\ref{fig:gathering}(b) reports the results of training in the Battle-Gathering game. Like the Multi Team Battle game, we plot the rewards obtained by Group A while fighting other groups for each algorithm. Again, MTMFQ shows the strongest performance in comparison to the other three algorithms.  The MFQ technique performs better than both MFAC and IL. In this game, MTMFQ converges in around 1500 episodes, while the other algorithms take around 1800 episodes to converge. The win rates shown in Figure~\ref{fig:gathering}(c) and the total rewards reported in Figure~\ref{fig:gathering}(d) also show the dominance of MTMFQ. 

The third domain is the multiagent Predator Prey (Figure~\ref{fig:predatorprey}(a)). Here we have two distinct types of agents --- predator and prey, each having completely different characteristics. The prey are faster than the predators and their ultimate goal is to escape the predators. The predators are slower than the prey, but they have higher attacking abilities than the prey. So the predators try to attack more and kill more prey. We model this game as an unknown type scenario where the types of the other agents in the competition are not known before hand (refer to Appendix B for more details). The MTMFQ algorithm plays the version with unknown types (Algorithm \ref{alg:MTMFQ-2}). Here we have four groups with the first two groups (Groups A and B) being predators and the other two groups (Groups C and D) being prey. Each algorithm will train two kinds of predator agents and two kinds of prey agents. All these agents are used in the playoff stage. In the playoff stage we have 800 games where we change the algorithm of predator and prey at every 200 games to maintain a fair competition. For the first 200 games, MTMFQ plays Group A, MFQ plays Group B, IL plays Group C, and MFAC plays Group D. In the next 200 games, MFAC plays Group A, MTMFQ plays Group B, MFQ plays Group C and IL plays Group D, and so on. We start all training and testing episodes with 90 prey and 45 predators for each group. Winning a game in the playoff stage is defined in the same way as the previous two games. Notice that this makes it more fair, as predators have to kill a lot more prey to win the game (as we start with more prey than predators) and prey have to survive longer. In this setup, the highly different types of agents make type identification easier for MTMFQ (as the types are initially unknown). The prey execute more move actions while the predators execute more attack actions. This can be well differentiated by clustering.

The results of the first training stage are reported in Figure~\ref{fig:predatorprey}(b). MTMFQ has comparable or even weaker performance than other algorithms in the first 600 episodes of the training and the reasoning is similar to the reasoning in the Multi Team Battle game (the agent strategies are not sufficiently differentiated for multiple types to be useful). Notice that for this game, MTMFQ takes many more episodes than the earlier two games to start dominating. This is because of the inherent hardness of this domain compared to the other domains (very different and unknown types). Similar to observations in the other domains, MTMFQ converges earlier (after around 1300 episodes as opposed to 1700 for the other algorithms). This shows its robustness to the different kinds of opponents in the environment. MTMFQ also gains higher cumulative rewards than the other algorithms. Win rates in Figure~\ref{fig:predatorprey}(c) show that MTMFQ still wins more games than the other algorithms, but the overall percentage of games won is less than the other domains. Thus, irrespective of the difficulty of the challenge we can see that MTMFQ has an upper hand. The lead of MTMFQ is also observed in Figure~\ref{fig:predatorprey}(d).

\section{Conclusion}

In this paper, we extended the notion of mean field theory to multiple types in MARL. We demonstrate that reducing many agent interactions to simple two agent interactions does not give very accurate solutions in environments where there are clearly different teams/types playing different strategies. We perform suitable experiments using MAgents and demonstrate superior performances using a type based approach. We hope that this paper will provide a different dimension to the mean field theory based MARL research.

 One limitation of our approach is that it is computationally more expensive than the mean field reinforcement learning method without types. If we really have only one type in the environment, then our method would add more compute time and not necessarily produce a better result. As future work we would like to extend this work for completely heterogeneous agents with different action spaces as well. StarCraft is one example of such a domain. Our work would be well suited for this scenario as clustering would be even easier. Another approach would be to consider sub types, further dividing types.

%%%%%%%%%%%%%%%%%%%%%%%%%%%%%%%%%%%%%%%%%%%%%%%%%%%%%%%%%%%%%%%%%%%%%%%%%%%%%%%%%%%%%%%%%%%%%%%%%%%%%%%%%
%% bibliography: see CFP for number of permitted pages

\bibliographystyle{ACM-Reference-Format}
\bibliography{bibfile}

 \newpage
\clearpage

\section*{Appendix A: Illustrative Example}

This section provides an example of a simulated scenario where the Multi Type Mean Field algorithm is more useful than the simple mean field algorithm.

Consider a game in which the central agent has to decide the direction of spin. The domain is stateless. The spin direction is influenced by the direction of spin of its A,B,C and D neighbours (four neighbours in total). Here we denote A as the neighbour to the left of the agent, B as the neighbour to the top of the agent, C as the neighbour to the right of the agent, and D as the neighbour to the bottom of the agent. The neighbour agents spin in one direction at random. If the agent spins in the same direction as both of its C and D neighbours, the agent gets a reward of -2 regardless of what the A and B are doing. If the spin is in the same direction as both of its A and B neighbours, the agent gets a +2 reward unless the direction is not the same as the one used by both of its C and D neighbours. All other scenarios result in a reward of 0. So the agent in Grid A in Figure~\ref{fig:counterexampleimages} will get a -2 for the spin down (since the C and D neighbours are spinning down) and a +2 for a spin up (since the A and B neighbours are spinning up). In Grid B of Figure~\ref{fig:counterexampleimages} the agent will get a -2 for spin up and a +2 for spin down. It is clear that the best action in Grid A is to spin up and the best action in Grid B is to spin down.

Here we have a notion of multi stage game with many individual stages. In each stage of a multi stage game, one or more players take one action each, simultaneously and obtain rewards.

Consider a sequence of stage games in which the agent gets Grid A for every 2 consecutive stages and then the Grid B for the third stage. The goal of the agent is to take the best possible action at all stages. Let us assume that the agent continues to learn at all stages, and it starts with Q values of 0. We apply MFQ (Equation \ref{eq:MFQ}) and MTMFQ (Equation \ref{eq:MTMFQ}) and show why MFQ fails, but MTMFQ succeeds in this situation. We are going to calculate the Q values for the 3 stages using both the MFQ (from \cite{pmlr-v80-yang18d}) and the MTMFQ update rules. We approximate the average action using the number of times the neighbourhood agents spin up. In the MTMFQ we use the A, B neighbours as the type 1 and the C, D neighbours as the type 2.

Applying MFQ: 

In the first stage, 

$$Q_1^j(\uparrow, \overline{a}^j =2) = 0 + 0.1(2-0) = 0.2 $$

$$Q_1^j(\downarrow, \overline{a}^j =2) = 0 + 0.1(-2-0) = -0.2 $$

Thus, the agent will choose to spin up in the first stage (correct action). 

For the second stage, 
$$Q_2^j(\uparrow, \overline{a}^j =2) = 0.38 $$

$$Q_2^j(\downarrow, \overline{a}^j =2) = -0.38 $$

Again the agent will choose to spin up in the second stage (correct action). 

For the third stage, 
$$Q_3^j(\uparrow, \overline{a}^j =2) = 0.38 + 0.1(-2-0.38) = 0.142 $$

$$Q_3^j(\downarrow, \overline{a}^j =2) = -0.38 + 0.1(2+ 0.38) =  -0.142 $$

Here again the agent will choose to make the spin up (wrong action). 

Now coming to MTMFQ updates, for the first stage, 

$$Q_1^j(\uparrow, \overline{a}^{j}_1 = 2, \overline{a}^{j}_2 = 0) = 0 + 0.1(2-0) = 0.2 $$

$$Q_1^j(\downarrow, \overline{a}^{j}_1 = 2, \overline{a}^{j}_2 = 0) = 0 + 0.1(-2-0) = -0.2 $$

Here the agent will spin up (correct action). 

For the second stage, 

$$Q_2^j(\uparrow, \overline{a}^{j}_1 = 2, \overline{a}^{j}_2 = 0) = 0.38 $$

$$Q_2^j(\downarrow, \overline{a}^{j}_1 = 2, \overline{a}^{j}_2 = 0) = -0.38 $$

Again the agent will spin up (correct action). 

For the third stage, 

$$Q_3^j(\uparrow, \overline{a}^{j}_1 = 0, \overline{a}^{j}_2 = 2) = -0.2 $$

$$Q_3^j(\downarrow, \overline{a}^{j}_1 = 0, \overline{a}^{j}_2 = 2) = 0.2 $$

Now it can be seen that the agent will spin down in this case (correct action). 

Thus, the MFQ agent will make one wrong move out of 3 moves whereas the MTMFQ agent will make the right move all the time. In situations like these, where the relationship of the agent with different neighbour agents is different, the MFQ algorithm would fail. The differences would be captured by MTMFQ which would take more efficient actions. This shows an example where MTMFQ outperforms MFQ.

\begin{figure}
\subfloat[Grid A]
{{
	\begin{tikzpicture}
	\draw (0, 0) grid (3, 3);
	\draw[very thick, scale=1] (1, 2) grid (2, 1);
    \draw [->,, red, very thick](0 +1.5 , 0.9) --
              (0 +1.5, 0.1);
    \draw [->,, red, very thick](0 +0.5 , 1.1) --
              (0 +0.5, 1.9);
    \draw [->,, red, very thick](0 +1.5 , 2.1) --
              (0 +1.5, 2.9);
    \draw [->,, red, very thick](0 +2.5 , 1.9) --
              (0 +2.5, 1.1);
              
    \draw (1.5,1.5) circle (8pt);

    \end{tikzpicture}
}}
\subfloat[Grid B]
{{
	\begin{tikzpicture}
	\draw (0, 0) grid (3, 3);
	\draw[very thick, scale=1] (1, 2) grid (2, 1);
    \draw [->,, red, very thick](0 +1.5 , 0.1) --
              (0 +1.5, 0.9);
    \draw [->,, red, very thick](0 +0.5 , 1.9) --
              (0 +0.5, 1.1);
    \draw [->,, red, very thick](0 +1.5 , 2.9) --
              (0 +1.5, 2.1);
    \draw [->,, red, very thick](0 +2.5 , 1.1) --
              (0 +2.5, 1.9);
     \draw (1.5,1.5) circle (8pt);

    \end{tikzpicture}
}}
  \caption{A counter example to show the failure of MFQ and success of MTMFQ.}
	\label{fig:counterexampleimages}
\end{figure}
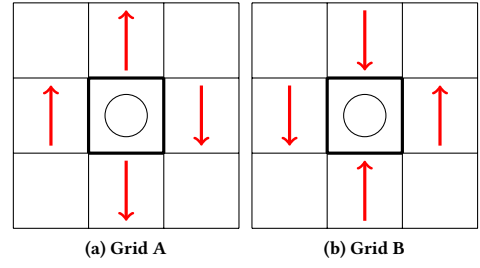

\section*{Appendix B: Experimental Details}

This section gives more details about the experimental conditions, especially the reward function.

For the Multi Team Battle domain, the agents in all of these groups get a reward of -0.01 for each move action, and a reward of -1 for dying. Attacking an empty grid has a reward of -0.1. Each of these members has a power of 10 and a damage of 2. When an agent loses all its powers, it dies. The power is like health that the agents maintain which gets depleted on being attacked. The agents can also recover (regain health points) from the attack using a step recovery rate. This is set to be 0.1. The positive rewards for attacking another agent is cyclic in nature with each group preferring to attack and kill a particular opponent group more than others. Group A has a positive reward of 0.2, 0.3, and 0.4 for attacking an agent of group B, C and D respectively, and it gets a reward of 80, 90, and 100 for killing a member of group B, C, and D respectively. Here we will call Group D as the favourable opponent of Group A as A gets more returns for fighting or killing D. The group B has a positive reward of 0.2, 0.3, and 0.4 for attacking a member of group C, D, and A respectively and a positive reward of 80, 90, and 100 for killing a member of Group C, D, and A respectively. Thus, Group A is the favourable opponent of Group B. Group C has a reward of 0.2, 0.3, and 0.4 for attacking Groups D, A, and B respectively. It gets a kill reward of 80, 90, and 100 for killing agents of Group D, A and B respectively. For Group D the order is Group A, B and C with a attack reward of 0.2, 0.3, and 0.4 and kill reward of 80, 90, and 100. 

For the Battle-Gathering game, the reward function is similar with an addition of each agent getting a +80 for collecting a food resource. The food resources are stationary objects, but each food resource has a power similar to agents. This power has to be reduced by damage before the food can be captured. Capturing a food will constitute making repeated efforts to gain the food resource by attacking the grid containing food. The agents would get a +0.5 for making such an attack.

For the Predator Prey domain, the predators have a power of 10 and a speed of 2 with an attack range of 2 (they can attack within a distance of 2 units) and they get a dead penalty of -0.1 and an attack penalty of -0.2. The step recovery rate is 0.1. The prey have a faster speed of 2.5 and an attack range of 0. All agents get a reward of +5 for killing agents belonging to other groups. Additionally, Group A gets a reward of +0.5 for attacking a member of Group B (since Group B is also a predator, Group A does not prefer to attack that), +1 for attacking a member of Group C, +3 for attacking a member of Group D (though both Groups C and D are prey, A prefers D to C).  Group B gets 0.5 for attacking A, but gets +1 for attacking D and +3 for attacking C. Group B prefers Group C to Group D. Every attack action entails a punishment for the (attack) receiver which is equal in value to the reward for the attacker.

\section*{Appendix C: Proof of Lemma \ref{lemma:jakkollalemma} used in this paper}

This section explains the changed Lemma used in this paper compared to Theorem 1 in \cite{jaakkola1994convergence}.

We state and prove a general theorem for a stochastic process before we begin the proof. 

\begin{theorem}\label{theorem:Ztransform}
If we have a stochastic process of the form 
\begin{equation}\label{eq:stochasticprocess}
    \begin{array}{l}
         y_{n+1} - py_{n} = d 
    \end{array}{}
\end{equation}{}
where $n$ goes from $0$ to $\infty$, then the general solution of this process can be given by 

\begin{equation}
    \begin{array}{l}
         \textrm{ if } p = 1: \\ \\
         
         y_n = dn + a;  \\ \\

         \textrm{ if } p \neq 1: \\ \\
         
         y_n = \frac{d}{1-p} + (a - \frac{d}{1-p}) p^n,

    \end{array}{}
\end{equation}{}
where $y_0 = a$.
\end{theorem}

\begin{proof}

Consider the expression, 

\begin{equation}
    \begin{array}{l}
         y_{n+1} - py_{n} = d.
          \end{array}{}
\end{equation}{}

Z transforms will be applied to solve this expression. Z transform is a generalized version of Discrete Time Fourier Transform that transforms the variables into a new subspace where solving the equation is easier. Once we obtain a solution, we can apply inverse Z transforms to get the solution in terms of the original variables.

\begin{equation}\label{eq:ztransform}
    \begin{array}{l}
         
    zY(z) - zy_0 - pY(z)  = \frac{dz}{z-1}
         \\ \\

         (z-p)Y(z) = \frac{dz}{z-1} + za
         
         \\ \\
         
     Y(z) = \frac{dz}{(z-1)(z-p)} + \frac{az}{z-p}.

          \end{array}{}
\end{equation}{}

Considering $p \neq 1$ we can rewrite the Equation \ref{eq:ztransform} as

\begin{equation}
    \begin{array}{l}
      Y(z) = dz[\frac{1}{(1-p)(z-1)} - \frac{1}{(1-p)(z-p)}] + \frac{az}{z-p}
          \\ \\
          
         Y(z) = \frac{d}{1-p} \frac{1}{1-z^{-1}} - \frac{d}{1-p}\frac{1}{1-pz^{-1}} + \frac{a}{1-pz^{-1}}
         
         \\ \\ 
         
         Y(z) = \frac{d}{1-p} \frac{1}{1-z^{-1}} + (a - \frac{d}{1-p})(\frac{1}{1-pz^{-1}}).

    \end{array}{}
\end{equation}{}

Now taking the inverse Z transform 

\begin{equation}
    \begin{array}{l}
         y_n =  [\frac{d}{1-p} + (a-\frac{d}{(1-p)}) p^n].
    \end{array}{}
\end{equation}{}

The above result is for the case where $p\neq1$. If $p=1$, see that the Equation \ref{eq:stochasticprocess} forms an arithmetic progression whose general term is $a+nd$.

\end{proof}{}

Now from Theorem \ref{theorem:Ztransform}, notice that if we want a general stochastic process of that form to converge we need the coefficient $p$ to be a fraction (as we take a limit to $\infty$ in the solution only a $p$ which is a fraction will give a converged result). Note that this result only needs $d$ to be finite, and it can be any finite number. If we iterate to infinity then we can apply the limit to the solution which will converge to $y_n = \frac{d}{1-p}$.

\begin{lemm}\label{lemma:appenlemma1}

A random process
\begin{equation} \label{eq:applemma}
    \begin{array}{l}
    w_{n+1}(x) = (1 - \alpha_n(x))w_n(x) + \beta_n(x) r_n(x)
    \end{array}
\end{equation}

\noindent converges to zero w.p.1, if the following conditions are satisfied: 
\begin{equation}
    \begin{array}{l}
         1) \hspace{10mm} \sum_n \alpha_n(x) = \infty, \sum_n \alpha^2_n(x) < \infty, \\  
         \hspace{13mm} \sum_n\beta_n(x) = \infty \textrm{ \textit{and} } \sum_n \beta^2_n(x) < \infty \\
         \hspace{13mm} \\
         \textrm{uniformly over x w.p.1} \\
        \\         
         2) \hspace{10mm} \E \{r_n(x) | P_n, \beta_n\} = 0 \\ 
         
         \hspace{13mm} \textrm{  \textit{and } }  \E\{r_n^2(x) |P_n, \beta_n\} \leq C \hspace{13mm}\\
         \\ 
         \textrm{w.p.1}, \textrm{where} \\
         
         P_n = \{w_n, w_{n-1}, \ldots, r_{n-1}, r_{n-1}, \ldots, \\ \alpha_{n-1}, \alpha_{n-2}, \ldots, \beta_{n-1}, \beta_{n-2}, \ldots\}
         \\
         
    \end{array}{}
\end{equation}

All the random variables are allowed to depend on the past $P_n$. $\alpha_n(x)$ and $\beta_n(x)$ are assumed to be non-negative and mutually independent given $P_n$. 

\end{lemm}{}

\begin{proof}
This is the same as Lemma 1 in \cite{jaakkola1994convergence}. The proof is based on that fact that we can divide the process $w_{n+1}$ (both sides of Equation~\ref{eq:applemma}) by a large value $W(x)$ such that $r_n(x) \ll W(x)$. Now the Equation~\ref{eq:applemma} is effectively reduced to

\begin{equation} \label{eq:applelemma2}
    \begin{array}{l}
    w_{n+1}(x) = (1 - \alpha_n(x))w_n(x).
    \end{array}
\end{equation}

\noindent
This is because of the conditions 1 and 2 which guarantees that $\beta_n$ is a fraction and that the variance of the process $r_n$ is finite. 

Now if you consider Equation \ref{eq:applelemma2}, the update is in such a way that the process is equal to a fraction of its previous value. Thus, this process converges to 0 w.p.1.

\end{proof}

\begin{lemm}\label{lemma:appenlemma2}
Consider the stochastic iteration 
\begin{equation}
    \begin{array}{l}
    
    X_{n+1}(x) = G_n(X_n, Y_n,x)         
         
    \end{array}{}
\end{equation}

\noindent
where $G_n$ is a sequence of functions and $Y_n$ is a random process. Let $(\Omega, \mathscr{F}, \mathscr{P})$ be a probability space. If the following are satisfied: 

1) The process is scale invariant. That is w.p.1 for all $\omega \in \Omega$

\begin{equation}
\begin{array}{l}
     G(\beta X_n, Y_n(\omega), x) = \beta G(X_n, Y_n(\omega), x).
\end{array}
\end{equation}

2) If we can keep $||X_n||$ bounded by scaling the process then $X_n$ would converge to a constant $D$ w.p.1 under condition 1.

The original process will converge to a constant $D_1=\frac{D}{\beta_0}$
where $\beta_0$ is the scaling factor that was applied to $||X_n||$.
\end{lemm}

\begin{proof}
The intuition of the proof is that we have a process that starts at a value and its first difference reduces with time till the value stabilizes at a point. Now, if this process is invariant to scaling, we can start the process with a small value (after scaling by a small fraction $\beta_0$) and then we can select a constant over which the $||X_n||$ should not increase (we can scale the whole process if it goes above that constant). Now according to the second condition the bounded process should converge to a constant $D$. Here the relation is $D \leq C$. To show that the net effect of the corrections must stay finite w.p.1, note that if $||X_n||$ converges then for any $\epsilon>0$ there exists $M_{\epsilon}$ such that $||X_n|| \leq D \leq C$ for all $n>M_{\epsilon}$ with probability at least $1-\epsilon$. This implies that the norm of the original process does not go above $C$ after $M_{\epsilon}$. Thus, convergence of $||X_n||$ to constant $D$ would then imply that the scaled version of the original process converges to the same constant $D$ w.p.1 under the bound. Now if we remove the scaling factor the convergence point of the original process is $\frac{D}{\beta_0}$.
\end{proof}

\begin{lemm}\label{lemma:appenlemma3}
A stochastic process $X_n$ which is bounded by the relation
\begin{equation}\label{eq:applelemma3}
    \begin{array}{l}
      |X_{n+1}(x)| = (1 - \alpha)|X_n(x)| + \gamma \beta C_1 + K
    \end{array}
\end{equation}
 converges to a constant $D$ w.p.1 provided 

1) $x \in S$, where S is a finite set. 

2) $\sum_n \alpha = \infty, \sum_n \alpha^2 < \infty$,  
$\sum_n \beta = \infty$, $\sum_n \beta^2 <              \infty$, \\ $\E\{\beta|P_n\} \leq E\{\alpha|P_n\}$, uniformly over $x$ w.p.1, \\ 
     where 
     \begin{equation*}
         \begin{array}{l}
     P_n = \{w_n, w_{n-1}, \ldots, r_{n-1}, r_{n-1}, \ldots, \alpha, \beta\}
         \end{array}
     \end{equation*}
  and $\alpha$, $\beta$ and $\gamma$ are assumed to be non negative.

 3) $K$ is finite.

  4) $\gamma \in (0,1)$   
   
  5) The original process $X_n$ is scale invariant. 
     
     The convergence point of the original process will then be $D = \frac{K + \gamma \beta C_1}{\alpha \beta_0} $,
    where $\beta_0$ is the scaling factor applied to the original process. 
     
     \end{lemm}{}
     
  \begin{proof}
  
    This is simply an application of Lemma \ref{lemma:appenlemma2}. 
 
    According to condition 5, we have a process that is scale invariant. Let us assume that we applied a scaling factor of $\beta_0$ to that process to get the bound as in Equation \ref{eq:applelemma3}. 
    Now, let us consider the iterative process
    
      \begin{equation}\label{eq:lemma3bound2}
        \begin{array}{l}
             |X_{n+1}(x)| = (1 - \alpha)|X_n(x)| + \gamma \beta C_1 + K.
             
        \end{array}
    \end{equation}

Equation \ref{eq:lemma3bound2} is linear in $|X_n(x)|$ and will converge to a point by Theorem \ref{theorem:Ztransform} w.p.1, to some $X^*(x)$, where $||X^*|| \leq E_1 $, where $E_1$ is a finite arbitrary constant.  From Theorem \ref{theorem:Ztransform} we can see that Equation \ref{eq:lemma3bound2} converges to a constant and hence, Lemma \ref{lemma:appenlemma2} can be applied. 
    The convergence point will be $\frac{K + \gamma \beta C_1}{\alpha}$ from Theorem \ref{theorem:Ztransform} for Equation \ref{eq:lemma3bound2}. If we change the value of the bound $C_1$ then the convergence point will change accordingly. To get the convergence point of the original process, we reapply the scaling factor. Thus, we get that point to be $\frac{K + \gamma \beta C_1}{\alpha \beta_0}$.

  \end{proof}

\begin{theorem}
A random iterative process 
\begin{equation*}
    \begin{array}{l}
         \Delta_{n+1}(x) = (1 - \alpha)\Delta_n(x) + \beta F_n(x) 
         
    \end{array}{}
\end{equation*}

\noindent
converges to a constant $D$ w.p.1 under the following conditions: 

1) $x \in S$ , where S is a finite set. 

2) $\sum_n \alpha = \infty, \sum_n \alpha^2 < \infty$, and \\
    $\sum_n \beta = \infty$, $\sum_n \beta^2 < \infty$, and \\
     $\E\{\beta|P_n\} \leq E\{\alpha |P_n\}$, \textit{uniformly over} $x$ \textit{w.p.1}. 
     
    3) $||\E\{ F_n(x) | P_n, \beta  \} ||_w \leq \gamma ||\Delta_n||_w + K$, where $ \gamma \in (0,1)$ and $K$ is finite. 
    
    4) $\textbf{var}\{F_n(x) | P_n, \beta\} \leq C(1+ ||\Delta_n||_W)^2$, where C is some constant. 
    Here 
$$ P_n = \{X_n, X_{n-1}, \ldots, F_{n-1}, \ldots, \alpha, \beta\}$$ 
    stands for the past at step $n$. $F_n(x)$ is allowed to depend on the past. $\alpha$ and $\beta$ are assumed to be non negative. The notation $||.||$ refers to some weighted maximum norm. 

The value of this constant $D = \frac {\psi C_1 + \beta |K|}{\alpha \beta_0}$
where $\psi \in (0,1)$ and $C_1$ is the constant with which the iterative process is bounded. $\beta_0$ is the scaling factor that was applied to the original process.
\end{theorem}

\begin{proof}

Be defining $r_n(x) = F_n(x) - \E\{F_n(x)|P_n, \beta\}$ we can decompose the iterative process into two parallel processes given by 

\begin{equation}\label{eq:theorempf}
    \begin{array}{l}
         \delta_{n+1}(x) = (1 - \alpha)\delta_n(x) + \beta \E\{F_n(x)|P_n, \beta\}
         \\
         w_{n+1}(x) = (1 - \alpha) w_n(x) + \beta r_n(x)
         
    \end{array}
\end{equation}

\noindent
where $\Delta_n(x) = \delta_n(x) + w_n(x)$. Dividing both the sides of Equation~\ref{eq:theorempf} by $\beta_0$ for each $x$ and denoting $\delta'_n(x) = \delta_n(x)/\beta_0$, $w'_n(x) = w_n(x)/\beta_0$ and $r'_n(x) = r_n(x)/\beta_0$ we can bound the $\delta'_n$ process by condition 3. 

Now we can rewrite the equation pair from condition 3 as

\begin{equation}
    \begin{array}{l}
         |\delta'_{n+1}| \leq (1 - \alpha)|\delta'_n(x)| + \gamma \beta ||\hspace{2mm} |\delta' | + w'_n || + \beta |K|
         \\
         w'_{n+1}(x) = (1 - \alpha)w'_n(x) + \gamma \beta r'_n(x).

    \end{array}
\end{equation}

Let us assume that the $\Delta_n$ process stays bounded. Then the variance of $r_n'(x)$ is bounded by some constant $C$ and thereby $w'_n$ converges to zero w.p.1 according to Lemma \ref{lemma:appenlemma1}. Hence, there exists $M$ such that for all $n>m$, $ || w'_n|| < \epsilon$ with probability at least $1-\epsilon$. This implies that the $\delta'_n$ process can be further bounded by 

\begin{equation}\label{eq:appendbound}
    \begin{array}{l}
         |\delta'_{n+1}| \leq (1-\alpha)|\delta'_n(x)| + \gamma \beta || \delta'_n - \epsilon|| + \beta|K|
    \end{array}{}
\end{equation}

\noindent
with $\textrm{probability} > 1 - \epsilon$. If we choose $C$ such that $\gamma (C+ 1)/C \leq 1$ then for $||\delta'_n|| > C\epsilon$

\begin{equation}
    \begin{array}{l}
         \gamma || \delta'_n + \epsilon|| \leq \gamma (C + 1)/C || \delta'_n||.
    \end{array}
\end{equation}{}

Note that in the above relation we do not need the term $\frac{C+1}{C}$ to be less than 1. We only need the product of this term with $\gamma$ to be less that one. 
Let us represent $F = (C+1)/C $. Now rewriting Equation \ref{eq:appendbound} we get the following bound,

\begin{equation}
    \begin{array}{l}
         |\delta'_{n+1}| \leq (1-\alpha)|\delta'_n(x)| + \gamma \beta F || \delta'_n|| + \beta |K|.
    \end{array}{}
\end{equation}

Let us bound the norm by $C_1$. Then we get the bound as 

\begin{equation}\label{eq:deltaboundequation}
    \begin{array}{l}
         |\delta'_{n+1}| \leq (1-\alpha)|\delta'_n(x)| + \gamma \beta F C_1 + \beta |K|.
    \end{array}{}
\end{equation}

Let us assume that $\delta_n$ is scale invariant (we prove that below). Now we can apply Lemma \ref{lemma:appenlemma3} as this satisfies all the conditions of Lemma \ref{lemma:appenlemma3}.  The original process converges to a constant $D$ w.p.1. Again according to Lemma \ref{lemma:appenlemma3} this constant value is $D = \frac{\gamma \beta F C_1 + \beta |K|}{\alpha \beta_0}$ where $\beta_0$ is the factor with which the original process was scaled. Let us denote a new fraction $\psi = \gamma \beta F$. Thus, the value of $D = \frac {\psi C_1 + \beta |K|}{\alpha \beta_0}$. Now, this guarantees w.p.1 convergence of the original process under the boundedness assumption if that process is scale invariant. Let the constant $D_1 = \frac{\gamma \beta F C_1 + \beta |K|}{\alpha}$ be the point Equation~\ref{eq:deltaboundequation} converges to according to Theorem \ref{theorem:Ztransform}. Since, the value of $D_1$ is very small for a small $K$ (we will show that $K$ is small in the application) and a small $C_1$ (since $C_1$ is arbitrary, we can choose a small $C_1$), we can get $|\delta_n| \approx \delta_n$.

Now we prove the scale invariance condition same as Jaakola et al. \cite{jaakkola1994convergence}. By Condition 4, $r'_n(x)$ can be written as $(1 + || \delta_n + w_n||)s_n(x)$, where $E\{ s_n^2(x)|P_n  \} \leq C$. Let us now decompose $w_n$ as $u_n + v_n $  with 

\begin{equation}
    \begin{array}{l}
         u_{n+1}(x) = (1 - \alpha)u_n(x) + \gamma \beta || \delta'_n + u_n + v_n|| s_n(x) \\ 
         
         v_{n+1}(x) = (1-\alpha) v_n(x) + \gamma \beta s_n(x)

    \end{array}{}
\end{equation}{}

\noindent
and $v_n$ converges to zero w.p.1 by Lemma \ref{lemma:appenlemma1}. Again by choosing $C$ such that $\gamma(C+1)/C < 1$ we can bound the $\delta'_n$ and $u_n$ processes for $||\delta'_n + u_n|| > C\epsilon$. The pair $(\delta'_n, u_n)$ is then a scale invariant process whose bounded version was proven earlier to converge to $D$ w.p.1. This proves the w.p.1 convergence of the triple $\delta'_n, u_n, v_n$ bounding the original process. 

\end{proof}

\end{document}